\documentclass[11pt]{article}
\usepackage[utf8]{inputenc}
\usepackage[]{amsmath}
\usepackage{amsthm}
\usepackage[]{upgreek}
\usepackage{xcolor}
\usepackage{mathtools}
\usepackage[]{gensymb}
\usepackage{changepage}
\usepackage[]{graphicx}
\usepackage[]{nicefrac}
\usepackage{tikz-cd} 
\usepackage{pgfplots}
\usetikzlibrary{decorations.pathmorphing}
\usetikzlibrary{patterns}
\usetikzlibrary{shapes}
\usetikzlibrary{plotmarks}
\usepackage[]{textgreek}
\usepackage[]{url}
\usepackage{xcolor,cancel}
\usepackage{pdfpages}
\usepackage[inline]{enumitem}
\usepackage[]{float}
\usepackage{braket}
\usepackage{graphicx, amssymb}
\usepackage{mathrsfs}
\usepackage{geometry}
\usepackage{fancyhdr}
\usepackage{sectsty}
\usepackage{slashed}
\usepackage[toc,page]{appendix}
\usepackage{tocloft}
\usepackage{lscape}
\usepackage{bbm}
\usepackage{tcolorbox}
\usepackage[numbers]{natbib}
\newtheorem{theorem}{Theorem}[section]
\newtheorem{corollary}[theorem]{Corollary}
\newtheorem{lemma}[theorem]{Lemma}

\newtheorem*{remark}{Remark}
\newtheorem{prop}[theorem]{Proposition}
\newtheorem{definition}{Definition}[section]

\renewcommand{\phi}{\varphi}

  \title{	\normalsize \textsc{} 	
		 		\LARGE \textbf{\uppercase{The Gregory--Laflamme Instability of\\ the Schwarzschild Black String Exterior}}	
		}
		\author{
		Sam C. Collingbourne \footnote{scc73@cam.ac.uk}\\
		Department of Pure Mathematics and Mathematical Statistics, University of Cambridge
        }

\begin{document}
\maketitle
\begin{abstract}
\noindent In this paper, a direct rigorous mathematical proof of the Gregory--Laflamme instability~\cite{GL2} for the five-dimensional Schwarzschild black string is presented. Under a choice of ansatz for the perturbation and a gauge choice originally introduced in~\cite{Wiseman}, the linearised vacuum Einstein equation reduces to an ODE problem for a single function. In this work, a suitable rescaling and change of variables is applied which casts the ODE into a Schr\"odinger eigenvalue equation to which an energy functional is assigned. It is then shown by direct variational methods that the lowest eigenfunction gives rise to an exponentially growing mode solution which has admissible behavior at the future event horizon and spacelike infinity. After the addition of a pure gauge solution, this gives rise to a regular exponentially growing mode solution of the linearised vacuum Einstein equation in harmonic/transverse-traceless gauge.
\end{abstract}

\pagebreak
\tableofcontents
\pagebreak
\section{Introduction}
The main topic of this paper is the study of the stability problem for the Schwarzschild black string solution to the Einstein vacuum equation in~$5$ dimensions. In 1994, the work of Gregory--Laflamme~\cite{GL2} gave strong numerical evidence for the presence of an exponentially growing mode instability. This phenomenon has since been known as the Gregory--Laflamme instability. This work has been widely invoked in the physics community to infer instability of many higher dimensional spacetimes, for example, black rings, ultraspinning Myers--Perry black holes and black Saturns. The interested reader should consult~\cite{H1,horo} and references therein, as well as~\cite{HarveyBB} and \cite{HolW,PW} which give a general approach to stability problems. The purpose of the present paper is to provide a direct, self-contained and elementary mathematical proof of the Gregory--Laflamme instability of the $5D$ Schwarzschild black string. 
\subsection{Schwarzschild Black Holes, Black Strings and Black Branes}
The most basic solution to the vacuum Einstein equation
\begin{eqnarray}
\mathrm{Ric}_g=0\label{VE}
\end{eqnarray}
giving rise to the black hole phenomena is the Schwarzschild black hole solution~$(\mathrm{Sch}_n,g_s)$. It arises dynamically as the maximal Cauchy development of the following initial data: an initial hypersurface~$\Sigma_0=\mathbb{R}\times \mathbb{S}^{n-2}$, a first fundamental form (in isotropic coordinates)
\begin{eqnarray}
h_s=\Big(1+\frac{M}{2\rho^{n-3}}\Big)^{\frac{4}{n-3}}(d\rho\otimes d\rho+\rho^2\mathring{\slashed{\gamma}}_{n-2}),\qquad \rho\in (0,\infty)\cong\mathbb{R}
\end{eqnarray}
and second fundamental form~$K=0$, where~$\mathring{\slashed{\gamma}}_{n-2}$ is the metric on the unit~$(n-2)$-sphere~$\mathbb{S}^{n-2}$. This spacetime is asymptotically flat and spherically symmetric. The following Penrose diagram represents the causal structure of~$(\mathrm{Sch}_n,g_s)$ arising from this initial data, restricted to the future of~$\Sigma_0$.
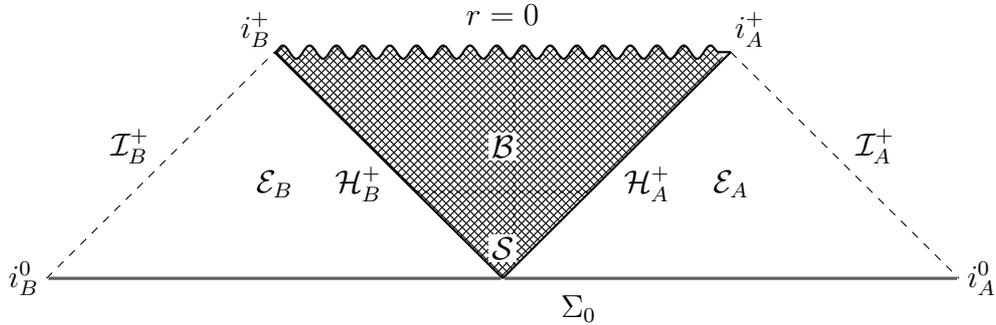
\begin{center}
\begin{tikzpicture}
     
      \draw [thick] (3,6) -- (6,3);
      \draw [thick] (6,3) -- (9,6);
      \draw [dashed] (9,6) -- (12,3);
       \draw [dashed] (0,3) -- (3,6);

      \draw [decorate, decoration=snake, thick] (3,6) -- (9,6);

      \fill [decoration={zigzag}]
        [pattern=north east lines,thick] (6,3) -- (9,6)
          decorate { (9,6) -- (3,6) } -- (6,3);
            \fill [decoration={zigzag}]
        [pattern=north west lines,thick] (6,3) -- (9,6)
          decorate { (9,6) -- (3,6) } -- (6,3);
    
      \node [inner sep=0pt] (A) at (0,3) {};
      \node [inner sep=0pt] (B) at (6,3) {};
      \node [inner sep=0pt] (C) at (12,3) {};
      \draw [very thick, gray!75!black] (A) to [bend right=0] (B);
      \draw [very thick, gray!75!black] (B) to [bend left=0] (C);

      \node (scriplus) at (10.9,4.75) {\large $\mathcal{I}^{+}_A$};
      \node (scriplus) at (1.1,4.75) {\large $\mathcal{I}^{+}_B$};

      \node (iplus) at (2.75,6.3) {\large $i^{+}_B$};
      \node (inaught) at (-.3,3) {\large $i^{0}_B$};
      \node (iplus) at (9.25,6.3) {\large $i^{+}_A$};
      \node (inaught) at (12.3,3) {\large $i^{0}_A$};

      \node (r0top) at (6,6.5) {\large $r=0$};

      \node [fill=white,inner sep=1pt](Hplus) at (7.9,4.25) {\large $\mathcal{H}^{+}_A$};   
      \node [fill=white,inner sep=1pt](Hminus) at (4.1,4.25) {\large $\mathcal{H}^{+}_B$};      

      \node [fill=white, inner sep=1pt] (Sigma) at (7,2.6) {\large $\Sigma_0$};
      \node [fill=white, inner sep=1pt] (Sigma) at (6,3.4) {\large $\mathcal{S}$};
      \node [fill=white, inner sep=1pt] (Int) at (6,4.75) {\large $\mathcal{B}$};
      \node [fill=white, inner sep=1pt] (Ext) at (9,4.25) {\large $\mathcal{E}_A$}  	;
            \node [fill=white, inner sep=1pt] (Ext) at (3,4.25) {\large $\mathcal{E}_B$}  	;
      \node[mark size=2pt] at (0,3) {\pgfuseplotmark{*}};
      \node[mark size=2pt] at (12,3) {\pgfuseplotmark{*}};
      \node[mark size=2pt] at (3,6) {\pgfuseplotmark{*}};
      \node[mark size=2pt] at (9,6) {\pgfuseplotmark{*}};
         \node[mark size=2pt] at (6,3) {\pgfuseplotmark{*}};
      \end{tikzpicture}   
\end{center}
\noindent To fix notation,~$\mathcal{I}^+:=\mathcal{I}_A^+\cup\mathcal{I}_B^+$ is future null infinity,~$i^+:=i^+_A\cup i^+_B$ and~$i^0:=i^0_A\cup i^0_B$ are future timelike infinity and spacelike infinity respectively,~$\mathcal{E}_A:=J^{-}(\mathcal{I}^+_A)\cap J^+(\Sigma_0)$ is the distinguished exterior region,~$\mathcal{E}_B:=J^{-}(\mathcal{I}^+_B)\cap J^+(\Sigma_0)$ is another exterior region,~$\mathcal{B}:=\mathrm{Sch}_n\setminus J^{-}(\mathcal{I}^+)$ is the black hole region,~$\mathcal{H}^+=\mathcal{H}^+_A\cup \mathcal{H}^+_B:=\mathcal{B}\setminus \mathrm{int}(\mathcal{B})$ is the future event horizon and~$\mathcal{S}:=\mathcal{H}_A^+\cap\mathcal{H}_B^+$ is the bifurcation sphere. The wavy line denotes a singular boundary which is not part of the spacetime~$(\mathrm{Sch}_n,g_s)$ but towards which the Kretchmann curvature invariant diverges. It is in this sense that~$(\mathrm{Sch}_n,g_s)$ is singular.  Note that every point in this diagram is in fact an~$(n-2)$-sphere. The metric on the exterior~$\mathcal{E}_A$ in traditional Schwarzschild coordinates~$(t,r,\phi_1,...,\phi_{n-2})$ takes the form
\begin{eqnarray}
g_s=-D_n(r)dt\otimes dt+\frac{1}{D_n(r)}dr\otimes dr+r^2\mathring{\slashed{\gamma}}_{n-2},\qquad D_n(r)=1-\frac{2M}{r^{n-3}},
\end{eqnarray}
where~$t\in[0,\infty)$,~$r\in \big((2M)^{\frac{1}{n-3}},\infty\big)$ and~$\mathring{\slashed{\gamma}}_{n-2}$ is the metric on the unit~$(n-2)$-sphere. \\

The Lorentzian manifold that is the main topic of this paper is the Schwarzschild black string spacetime in~$5$ dimensions which is constructed from the~$4D$ Schwarzschild solution $(\mathrm{Sch}_4,g_s)$. Before focussing on this spacetime explicitly, it is of interest to discuss more general spacetimes constructed from the~$n$-dimensional Schwarzschild black hole solution~$(\mathrm{Sch}_n,g_s)$. Let~$\mathbb{S}_R^1$ denote the circle of radius~$R$ and let~$\mathrm{F}_p\in\{\mathbb{R}^p,\mathbb{R}^{p-1}\times \mathbb{S}^1_R,...,\mathbb{R}\times \prod_{i=1}^{p-1}\mathbb{S}_{R_i}^1,\prod_{i=1}^{p}\mathbb{S}_{R_i}^1\}$ with its associated~$p$-dimensional Euclidean metric~$\delta_p$. If one has the~$n$-dimensional Schwarzschild black hole spacetime~$(\mathrm{Sch}_n,g_s)$ and takes its Cartesian product with~$\mathrm{F}_p$ then one realises the~$(n+p)$-dimensional Schwarzschild black brane~$(\mathrm{Sch}_n\times \mathrm{F}_p,g_s\oplus \delta_p)$. This means that the~$(n+p)$-dimensional Schwarzschild black brane~$(\mathrm{Sch}_n\times \mathrm{F}_p,g_s\oplus \delta_p)$ is a product manifold made from Ricci-flat manifolds, which is again Ricci-flat and hence satisfies the vacuum Einstein equation~(\ref{VE}). Note that in contrast to~$(\mathrm{Sch}_n,g_s)$, the spacetimes~$(\mathrm{Sch}_n\times \mathrm{F}_p,g_s\oplus \delta_p)$ are not asymptotically flat but are called `asymptotically Kaluza--Klein'. \\

The Schwarzschild black brane spacetimes~$(\mathrm{Sch}_n\times \mathrm{F}_p,g_s\oplus \delta_p)$ arise dynamically as the maximal Cauchy development of suitably extended Schwarzschild initial data, i.e.,~$(\Sigma_0\times \mathrm{F}_p,h_s\oplus\delta_p, K=0)$. Hence, the above Penrose diagram can be reinterpreted as the Penrose diagram for the Schwarzschild black brane, but instead of each point representing a~$(n-2)$-sphere, it represents a~$\mathbb{S}^{n-2}\times \mathrm{F}_p$. In particular, the notation~$\mathcal{E}_A$ will be used henceforth to denote the distinguished exterior region of~$(\mathrm{Sch}_n\times \mathrm{F}_p,g_s\oplus \delta_p)$.  \\

Taking~$p=1$ gives rise to the~$(n+1)$-dimensional Schwarzschild black string spacetime~$\mathrm{Sch}_n\times \mathbb{R}$ or alternatively~$\mathrm{Sch}_n\times \mathbb{S}^1_R$. The topic of the present paper is the~$5D$ Schwarzschild black string spacetime~$\mathrm{Sch}_4\times \mathbb{R}$ or alternatively~$\mathrm{Sch}_4\times \mathbb{S}_R^1$. The metric on the exterior~$\mathcal{E}_A$ in standard Schwarzschild coordinates is
\begin{eqnarray}
g:=-D(r)dt\otimes dt+\frac{1}{D(r)}dr\otimes dr+r^2\mathring{\slashed{\gamma}}_{2}+dz\otimes dz,\qquad D(r)=1-\frac{2M}{r}\label{BSM}
\end{eqnarray}
where $t\in[0,\infty)$, $r\in (2M,\infty)$ and $z\in\mathbb{R}$ or $\nicefrac{\mathbb{R}}{2\pi R\mathbb{Z}}$.\\

Finally, to analyse the subsequent problem of linear stability on the exterior region~$\mathcal{E}_A$ up to the future event horizon $\mathcal{H}_A^+$, one requires a chart with coordinate functions that are regular up to this hypersurface~$\mathcal{H}^+_A\setminus \mathcal{S}$, where~$ \mathcal{S}$ now denotes the bifurcation surface. A good choice is ingoing Eddington--Finkelstein coordinates defined by
\begin{eqnarray}
v=t+r_*,\qquad \frac{dr_*}{dr}=\frac{r^{n-3}}{r^{n-3}-2M},\quad \text{with } r_*(3M)=3M+2M\log(M).
\end{eqnarray}
The~$(n+p)$-dimensional Schwarzschild black brane metric becomes
\begin{eqnarray}
g_s\oplus\delta=-D_n(r)dv\otimes dv+dv\otimes dr+dr\otimes dv+r^2\mathring{\slashed{\gamma}}_{n-2}+\delta_{ij}dz^i\otimes dz^j,\qquad D_n(r)=1-\frac{2M}{r^{n-3}}.
\end{eqnarray}
\subsection{Previous Works}
For a good introduction to the Gregory--Laflamme instability and the numerical result of~\cite{GL2} see the book chapter~\cite{GL3}. A detailed survey of the key work~\cite{GL4} related to the present paper is undertaken in section~\ref{spherical}. A brief history of the problem is presented here:
\begin{enumerate}
\item In 1988, Gregory--Laflamme examined the Schwarzschild black string spacetime and stated that it is stable~\cite{GL1}. However, an issue in the analysis arose from working in Schwarzschild coordinates which lead to incorrect regularity assumptions for the asymptotic solutions. 
\item In 1993, Gregory--Laflamme used numerics to give strong evidence for the existence of a low-frequency instability of the Schwarzschild black string and branes in harmonic gauge~\cite{GL2}. 
\item In 1994, Gregory--Laflamme generalised their numerical analysis to show instability of `magnetically-charged dilatonic' black branes~\cite{GL5} (see \cite{GL5,HorStr} for a discussion of these solutions). 
\item In 2000, Gubser--Mitra discussed the Gregory--Laflamme instability for general black branes. They conjectured that a necessary and sufficient condition for stability of the black brane spacetimes is thermodynamic stability of the corresponding black hole~\cite{GM1,GM2}.
\item In 2000, Reall~\cite{HarveyE}, with the aim of addressing the  Gubser--Mitra conjecture, explored further the relation between stability of black branes arising from static, spherically symmetric black holes and thermodynamic stability of those black holes. In particular, the work of Reall argues that there is a direct relation between the `negative mode' of the Euclidean Schwarzschild instanton solution (this mode was initially identified in a paper by Gross, Perry and Yaffe~\cite{GPY}) and the threshold of the Gregory--Laflamme instability. This idea was further explored in a work of Reall et al.~\cite{HarveyE2}, which extended the idea that `negative modes' of the Euclidean extension of a Myers--Perry black hole\footnote{The Myers--Perry black hole is the generalisation of the Kerr spacetime to higher dimensions, see~\cite{MP,H1} for details.} correspond to the threshold for the onset of a Gregory--Laflamme instability. 
\item In 2006, Hovdebo and Myers~\cite{GL4} used a different gauge (which was introduced in~\cite{Wiseman}) to reproduce the numerics from the original work of Gregory and Laflamme. This gauge choice will be called spherical gauge and will be adopted in the present work. This work discusses the presence of the Gregory--Laflamme instability for the `boosted' Schwarzschild black string and the Emparan--Reall black ring (for a discussion of this solution see~\cite{Harvey,H1,H2}).
\item In 2010, Lehner and Pretorius numerically simulated the non-linear evolution of the Gregory--Laflamme instability; see the review~\cite{NLGL} and references therein. 
\item In 2011, Figueras, Murata and Reall~\cite{HarveyBB} put forward the idea that a local Penrose inequality gives a stability criterion. Furthermore,~\cite{HarveyBB} showed numerically that this local Penrose inequality was violated for the Schwarzschild black string for a range of frequency parameters which closely match those found in the original work of Gregory--Laflamme~\cite{GL2}.
\item In 2012, Hollands and Wald~\cite{HolW} and, later in 2015, Prabu and Wald~\cite{PW} developed a general method applicable to many linear stability problems which encompasses the problem of linear stability of the Schwarzschild black string exterior $\mathcal{E}_A$. The papers~\cite{HolW} and~\cite{PW} are explored in detail in section~\ref{CEM}.
\end{enumerate}

\subsection{Statement of the Main Theorem}
The purpose of this paper is to give a direct, self-contained, elementary proof of the Gregory--Laflamme instability for the $5D$ Schwarzschild black string. \\

For the statement of the main theorem, one should have in mind the following Penrose diagram for the~$5D$ Schwarzschild black string spacetime:
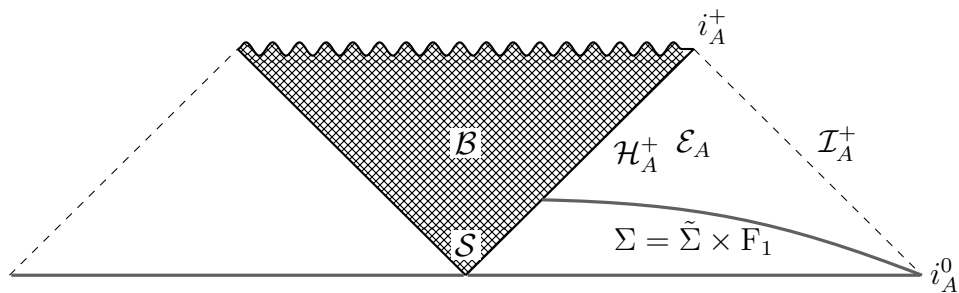
\begin{center}
\begin{tikzpicture}
     
      \draw [thick] (3,6) -- (6,3);
      \draw [thick] (6,3) -- (9,6);
      \draw [dashed] (9,6) -- (12,3);
       \draw [dashed] (0,3) -- (3,6);

      \draw [decorate, decoration=snake, thick] (3,6) -- (9,6);

      \fill [decoration={zigzag}]
        [pattern=north east lines,thick] (6,3) -- (9,6)
          decorate { (9,6) -- (3,6) } -- (6,3);
            \fill [decoration={zigzag}]
        [pattern=north west lines,thick] (6,3) -- (9,6)
          decorate { (9,6) -- (3,6) } -- (6,3);

      \node [inner sep=0pt] (A) at (0,3) {};
      \node [inner sep=0pt] (B) at (6,3) {};
      \node [inner sep=0pt] (C) at (12,3) {};
      \draw [very thick, gray!75!black] (A) to [bend right=0] (B);
      \draw [very thick, gray!75!black] (B) to [bend left=0] (C);
      \draw [very thick, gray!75!black] (7,4) to [bend left=10] (12,3);
      
      \node (scriplus) at (10.9,4.75) {\large $\mathcal{I}^{+}_A$};
   
      \node [fill=white, inner sep=1.5pt] (bifurcate) at (6,3.4) {\large $\mathcal{S}$};

      \node (iplus) at (9.25,6.3) {\large $i^{+}_A$};
      \node (inaught) at (12.3,3) {\large $i^{0}_A$};
      \node [fill=white,inner sep=1pt](Hplus) at (8.25,4.6) {\large $\mathcal{H}^{+}_A$};   
      \node  (sig) at (9,3.5) {\large $\Sigma=\tilde{\Sigma}\times \mathrm{F}_1$};
     
      \node [fill=white, inner sep=1.5pt] (Int) at (6,4.75) {\large $\mathcal{B}$};
      \node [fill=white, inner sep=1pt] (Ext) at (9,4.75) {\large $\mathcal{E}_A$}  	;
      \node[mark size=2pt] at (0,3) {\pgfuseplotmark{*}};
      \node[mark size=2pt] at (12,3) {\pgfuseplotmark{*}};
      \node[mark size=2pt] at (3,6) {\pgfuseplotmark{*}};
      \node[mark size=2pt] at (9,6) {\pgfuseplotmark{*}};
        \node[mark size=2pt] at (7,4) {\pgfuseplotmark{*}};
           \node[mark size=2pt] at (6,3) {\pgfuseplotmark{*}};
      \end{tikzpicture}   
\end{center}
In the above Penrose diagram,~$\tilde{\Sigma}$ is a spacelike asymptotically flat hypersurface which intersects the future event horizon~$\mathcal{H}_A^+$ to the future of the bifurcation surface~$\mathcal{S}$ and $\mathrm{F}_1=\mathbb{R}$ or $\mathbb{S}^1_R$. One should note that~$\Sigma$ can be expressed as~$\Sigma=\{(t,r_*,\theta,\phi,z):t=f(r_*)\}$ such that~$f\sim 1$ for~$r_*\rightarrow \infty$. An explicit example would be a hypersurface of constant~$t_*$ where
\begin{eqnarray}
t_*=t+2M\log(r-2M). 
\end{eqnarray}
\begin{definition}[Mode Solution]
A solution of the linearised vacuum Einstein equation
\begin{eqnarray}
g^{cd}\nabla_c\nabla_dh_{ab}+\nabla_a\nabla_bh-2\nabla_{(b}\nabla^{c}h_{a)c}+2{{{R_{a}}^{c}}_b}^dh_{cd}=0\label{PertE}
\end{eqnarray} 
on the exterior~$\mathcal{E}_A$ of the Schwarzschild black string~$\mathrm{Sch}_4\times \mathbb{R}$ of the form
\begin{eqnarray}
h_{\alpha\beta}=e^{\mu t+i\omega z}H_{\alpha\beta}(r,\theta)\label{ansatz}
\end{eqnarray}
with~$\mu,\omega\in \mathbb{R}$ and~$(t,r,\theta,\phi,z)$ standard Schwarzschild coordinates will be called a mode solution of~(\ref{PertE}).
\end{definition}
A way of establishing the linear instability of an asymptotically flat black hole is exhibiting a mode solution of the linearised Einstein equation~(\ref{PertE}) which is smooth up to and including the future event horizon, decays towards spacelike infinity and such that $\mu>0$. 
\vspace{-1em}
\begin{adjustwidth}{18pt}{18pt}
\begin{theorem}[Gregory--Laflamme Instability]\label{RT}
For all~$|\omega|\in [\frac{3}{20M},\frac{8}{20M}]$, there exists a non-trivial mode solution~$h$ of the form~(\ref{ansatz}) to the linearised vacuum Einstein equation~(\ref{PertE}) on the exterior~$\mathcal{E}_A$ of the Schwarzschild black string background~$\mathrm{Sch}_4\times \mathbb{R}$ with~$\mu>\frac{1}{40\sqrt{10}M}>0$ and
\begin{eqnarray}
H_{\alpha\beta}(r,\theta)=\begin{pmatrix}
H_{tt}(r)&H_{tr}(r)&0&0&0\\
H_{tr}(r)&H_{rr}(r)&0&0&0\\
0&0&H_{\theta\theta}(r)&0&0\\
0&0&0&H_{\theta\theta}(r)\sin ^2\theta&0\\
0&0&0&0&0
\end{pmatrix}.\label{matrixconstruct}
\end{eqnarray}
The solution~$h$ extends regularly to~$\mathcal{H}_A^+$ and decays exponentially towards~$i_A^0$ and can thus be viewed as arising from regular initial data on a hypersurface~$\Sigma$ extending from the future event horizon~$\mathcal{H}^+_A$ to~$i^0_A$. In particular,~$h|_{\Sigma}$ and~$\nabla h|_{\Sigma}$ are smooth on~$\Sigma$. Moreover, the solution~$h$ is not pure gauge and can in fact be chosen such that the harmonic/transverse-traceless gauge conditions
\begin{eqnarray}
\begin{cases}
\nabla^ah_{ab}=0\label{HG}\\
g^{ab}h_{ab}=0
\end{cases}
\end{eqnarray}
are satisfied.\\

Suppose~$R>4M$, then one can choose~$\omega$ such that there exists an integer~$ n\in [\frac{3R}{20M},\frac{8R}{20M}]$ and therefore~$h$ induces a smooth solution on the exterior~$\mathcal{E}_A$ of the Schwarzschild black string~$\mathrm{Sch}_4\times\mathbb{S}^1_R$. Moreover, the initial data for such a mode solution on the exterior~$\mathcal{E}_A$ of~$\mathrm{Sch}_4\times\mathbb{S}^1_R$ has finite energy.\\
 
Hence, the exterior~$\mathcal{E}_A$ of the Schwarzschild black string~$\mathrm{Sch}_4\times \mathbb{R}$ or~$\mathrm{Sch}_4\times \mathbb{S}^1_R$ for~$R>4M$ is linearly unstable as a solution of the vacuum Einstein equation~(\ref{PertE}), and the instability can be realised as a mode instability in harmonic/transverse-traceless gauge~(\ref{HG}) which is not pure gauge. 
\end{theorem}
\end{adjustwidth}
\begin{remark}
Since the mode solution constructed in Theorem~\ref{RT} is not pure gauge, one expects that the above mode solution persists in any `good' gauge, not just~(\ref{HG}). It would be of interest to also formalise this in terms of a gauge invariant quantity. 
\end{remark}

\subsection{Difficulties and Main Ideas of the Proof} \label{DMIP}
It may seem natural to directly consider the problem in harmonic gauge since the equation of study~(\ref{PertE}) reduces to a tensorial wave equation with an inhomogeneity
\begin{eqnarray}
g^{cd}\nabla_{c}\nabla_dh_{ab}+2{{{R_{a}}^{c}}_b}^dh_{cd}=0.\label{PertE2}
\end{eqnarray}
The above equation~(\ref{PertE2}) results from the linearisation of the gauge reduced non-linear vacuum Einstein equation~(\ref{VE}) which is strongly hyperbolic and therefore well-posed. The equation~(\ref{PertE2}) reduces to a system of ODEs under the mode solution ansatz~(\ref{ansatz}) with~(\ref{matrixconstruct}). If one wishes to reduce this system to a single ODE in $H_{tt}$, $H_{tr}$, $H_{rr}$ or $H_{\theta\theta}$ one introduces a regular singular point in the range~$r\in (0,\infty)$. For certain ranges of~$\mu$ and~$\omega$, this value occurs on the exterior~$\mathcal{E}_A$, i.e., the regular singular point occurs in~$r\in (2M,\infty)$.  In particular, this regular singularity occurs on the exterior for the numerical values of $\omega$ and $\mu$ for which Gregory--Laflamme identified instability. In the original works of Gregory and Laflamme the decoupled ODE for $H_{tr}
$ was studied; see the works~\cite{GL1,GL2,GL3}. \\

It turns out that, in looking for an instability one can make a different gauge choice called spherical gauge. As shown in section~\ref{spherical}, the linearised vacuum Einstein equation~(\ref{PertE}) for a mode solution~(\ref{ansatz}) in spherical gauge can be reduced to a~$2^{\mathrm{nd}}$-order ODE of the form
\begin{eqnarray}
\frac{d^2u}{dr^2}+P_{\omega}(r)\frac{du}{dr}+Q_{\omega}(r)u=\frac{\mu^2}{D(r)^2}u,\qquad D=1-\frac{2M}{r},\label{E1}
\end{eqnarray}
where~$P_{\omega}(r)$ and~$Q_{\omega}(r)$ depend only on~$\omega$ and~$r$. This was originally introduced in~\cite{Wiseman} and avoids the issues of a regular singularity at some~$r\in (2M,\infty)$. However, in contrast to harmonic gauge, for this gauge choice, well-posedness is unclear. If one were trying to prove \emph{stability} then exhibiting a well-posed gauge would be key since well-posedness of the equations is essential for understanding general solutions.  For \textit{instability}, it turns out that it is sufficient to exhibit a mode solution of the non-gauge reduced equation~(\ref{PertE}) which is \textit{not pure gauge}. One expects then that such a mode solution will persist in all `good' gauges, of which harmonic gauge is an example. The discussion of pure gauge mode solutions in spherical gauge in section~\ref{notpuregaugeS} produces the restriction that a mode solution in spherical gauge is \emph{not} pure gauge if~$\omega\neq 0$. \\

An issue with spherical gauge is that mode solutions in the spherical gauge do not, in general, extend smoothly to the future event horizon~$\mathcal{H}^+_A$, even when they represent physically admissible solutions. However, as shown in section~\ref{RC}, one can detect what are the admissible boundary conditions at the future event horizon in spherical gauge by adding a pure gauge perturbation to the metric perturbation to try and construct a solution that indeed extends smoothly to $\mathcal{H}^+_A$. In fact, the pure gauge perturbation found is precisely one that transforms the metric perturbation to harmonic/transverse-traceless gauge~(\ref{HG}). Hence, after also identifying the admissible boundary conditions at spacelike infinity~$i^0_A$ in section~\ref{RC}, proving the existence of an unstable mode solution to the linearised vacuum Einstein equation~(\ref{PertE}) that is \emph{not} pure gauge is reduced to showing the existence of a solution to the ODE~(\ref{E1}) with~$\mu>0$ and~$\omega\neq 0$ which satisfies the admissible boundary conditions that are identified in this work. \\
 
In this paper, the ODE problem~(\ref{E1}) is approached from a direct variational point of view in section~\ref{BSU}. To run a direct variational argument, the solution~$u$ of ODE~(\ref{E1}) is rescaled and change of coordinates is applied. It is shown in section~\ref{AODE} that equation~(\ref{E1}) can be cast into a Schr\"odinger form
\begin{eqnarray}
-\Delta_{r_*}u+V_{\omega}(r_*)u=-\mu^2u,\qquad r_*=r+2M\log(r-2M)\label{E2}
\end{eqnarray}
with~$V_{\omega}$ independent of~$\mu$. The ODE~(\ref{E2}) can be interpreted as an eigenvalue problem for~$-\mu^2$; finding an eigenfunction, in a suitable space, with a negative eigenvalue will correspond to an instability. As shown in section~\ref{DVA}, this involves assigning the following energy functional to the Schr\"odinger operator on the left-hand side of~(\ref{E2}):
\begin{eqnarray}
E(u):=\langle\nabla_{r_*} u,\nabla_{r_*} u\rangle_{L^2(\mathbb{R})}+\langle V_{\omega} u,u\rangle_{L^2(\mathbb{R})}.\label{E3}
\end{eqnarray}
Using a suitably chosen test function, one can show that the infimum over functions in~$H^1(\mathbb{R})$ of this functional is negative for a range of~$\omega$. One then needs to argue that this infimum is attained as an eigenvalue, by showing this functional is lower semicontinuous and that the minimizer is non-trivial. The corresponding eigenfunction is then a weak solution in~$H^1(\mathbb{R})$ to the ODE~(\ref{E2}) with~$\mu>0$ for a range of~$\omega\in \mathbb{R}\setminus \{0\}$. Elementary one-dimensional elliptic regularity implies the solution is indeed smooth away from the future event horizon,~$\mathcal{H}^+_A$, and therefore corresponds to a classical solution of the problem~(\ref{E2}). Finally, the solution can be shown to satisfy the admissible boundary conditions by the condition that the solution lies in~${H}^1(\mathbb{R})$. \\

The paper is organised in the following manner. The remainder of the present section contains additional background on the Gregory--Laflamme instability. In section~\ref{LPT}, linear perturbation theory is reviewed and the linearised Einstein equation~(\ref{PertE}) is derived. In section~\ref{spherical}, the analysis in spherical gauge is presented. The decoupled ODE~(\ref{E1}) resulting from the linearised Einstein equation~(\ref{PertE}) is derived and it is established that the problem can be reduced to the existence of a solution to the decoupled ODE with~$\mu>0$ and~$\omega\neq 0$ satisfying admissible boundary conditions. In section~\ref{BSU}, the proof of the existence of such a solution is presented via the direct variational method.\\

 Appendix~\ref{ChR} contains a list of the Riemann tensor components and the Christoffel symbols for the Schwarzschild black string spacetime~$\mathrm{Sch}_4\times \mathbb{R}$ or~$\mathrm{Sch}_4\times \mathbb{S}^1_R$. Appendix~\ref{Sing} collects results on singularities in 2nd order ODE relevant for the discussion of the boundary conditions for the decoupled ODE~(\ref{E1}). Appendix~\ref{SchT} provides a method of transforming a 2nd order ODE into a Schr\"odinger equation. Appendix~\ref{analysis} collects some useful results from analysis that are needed in the proof of theorem~\ref{RT}. Appendix~\ref{bounds} compliments theorem~\ref{RT} with some stability results. 

\subsection{The Canonical Energy Method}\label{CEM}
 The reader should note that there are two papers~\cite{HolW,PW} concerning a very general class of spacetimes which are of relevence to the stability problem for the Schwarzschild black string. In particular, it follows from~\cite{HolW,PW} that there exists a linear perturbation of the Schwarzschild black string spacetime which is not pure gauge and grows exponentially in the Schwarzschild~$t$-coordinate. The following describes the results of these works. \\

In 2012, a paper of Hollands and Wald~\cite{HolW} gave a criterion for linear stability of stationary, axisymmetric, vacuum black holes and black branes in~$D\geq 4$ spacetime dimensions under axisymmetric perturbations. They define a quantity called the `canonical energy'~$\mathcal{E}$ of the perturbation which is an integral over an initial Cauchy surface of an expression quadratic in the perturbation. It can be related to thermodynamic quantities by 
\begin{eqnarray}
\mathcal{E}=\delta^2M-\sum_B\Omega_B\delta^2J_B-\frac{\kappa}{8\pi}\delta^2A,\label{can}
\end{eqnarray}
where~$M$ and~$J_{B}$ are the ADM mass and ADM angular momenta in the~$B^{\mathrm{th}}$ plane, and~$A$ is the cross-sectional area of the horizon. Note that the right-hand side of~(\ref{can}) refers to the second variation of thermodynamic quantities. It is remarkable that the combination $\mathcal{E}$ of these second variations is in fact determined by linear perturbations.   \\

The work~\cite{HolW} considers initial data for a perturbation of either a stationary, axisymmetric black hole or black brane with the following properties:
\begin{enumerate*}
\item[(i)] the linearised Hamiltonian and momentum constraints are satisfied,
\item[(ii)] that~$\delta M=0=\delta J_A$ and that the ADM momentum vanishes and 
\item[(iii)] specific gauge  conditions and finiteness/regularity conditions at the future horizon and infinity are satisfied.
\end{enumerate*}
In what follows, initial data satisfying (i)--(iii) will be referred to as admissible. Hollands and Wald show that if~$\mathcal{E}\geq 0$ for all admissible initial data, then one has mode stability. The work also establishes that if there exist admissible initial data such that~$\mathcal{E}<0$, then there exist admissible initial data for a perturbation which cannot approach a stationary perturbation at late times, i.e., one has failure of asymptotic stability. \\

For the Schwarzschild black hole, one can take initial data which corresponds simply to a change of the mass parameter $M\mapsto M+\alpha$ and therefore, by equation~(\ref{can}) and since the cross-sectional area of the horizon is given by $A=16\pi(M+\alpha)^2$, it follows that~$\mathcal{E}<0$. This is the `thermodynamic instability' of the Schwarzschild black hole. However, the initial data for a change of mass perturbation is manifestly \underline{not} admissible (the family of Schwarzschild black holes is, after all, dynamically stable).\\

The work of Hollands and Wald~\cite{HolW} also shows an additional result relevant specifically to the problem of stability of black \emph{branes}. Suppose there exist initial data for a perturbation of the ADM parameters of a black \emph{hole} such that $\mathcal{E}<0$. The work \cite{HolW} shows that, starting from such a perturbation of the black \emph{hole}, one can infer the existence of \underline{admissible} initial data, which depend on a parameter $l$, for a perturbation (which is not pure gauge) of the associated black \emph{brane} such that again $\mathcal{E}<0$. One should note that this argument does not give an explicit bound on $l$. This criterion formalised a conjecture by Gubser--Mitra that a necessary and sufficient condition for stability of the black brane spacetimes is thermodynamic stability of the corresponding black hole~\cite{GM1,GM2}. Since the change of mass perturbation of Schwarzschild black hole produces $\mathcal{E}<0$, this argument implies that the Schwarzschild black string fails to be asymptotically stable.  

\begin{remark} The reader should note that the Hollands and Wald paper~\cite{HolW} also showed that a necessary and sufficient condition for stability, with respect to axisymmetric perturbations, is that a `local Penrose inequality' is satisfied. The idea that a local Penrose inequality gives a stability criterion was originally discussed in the work of Figueras, Murata and Reall~\cite{HarveyBB} which gave strong evidence in favor of sufficiency of this condition for stability. Furthermore~\cite{HarveyBB} showed numerically that this local Penrose inequality was violated for the Schwarzschild black string for a range of frequency parameters which closely match those found in the original work of Gregory--Laflamme~\cite{GL2}.
\end{remark} 
The failure of asymptotic stability does not in itself imply that perturbations grow. However, the results of~\cite{HolW} were strengthened in 2015 by Prabhu and Wald~\cite{PW}. They showed, using some spectral theory, that if there exist admissible initial data for a perturbation such that~$\mathcal{E}<0$ for a black \emph{brane}, then there exists initially well-behaved perturbations that are not pure gauge and that grow exponentially in time. Having established that there exist admissible initial data for a perturbation such that~$\mathcal{E}<0$ for the Schwarzschild black string in~\cite{HolW}, existence of a linear perturbation which is not pure gauge and has exponential growth follows.\\

The present work differs from the above as it gives a direct, self-contained, elementary proof of the Gregory--Laflamme instability following the original formulation of~\cite{GL1,GL2,GL3,GL4} which is completely explicit. In particular, it gives an exponentially growing mode solution with an explicit growth rate, of the form defined by equations~(\ref{ansatz}) and~(\ref{matrixconstruct}) in harmonic/transverse-traceless gauge which is not pure gauge.
\begin{remark}
It would also be of interest to see if Theorem~\ref{RT} in the form stated could be inferred from the canonical energy method of Hollands, Wald and Prabu~\cite{HolW,PW} in an explicit way bypassing some of the functional calculus applied there. In particular, it would be interesting to explore the possible relation between the variational theory applied to $\mathcal{E}$ and that applied here (see section~\ref{DVA}). 
\end{remark}
\subsection{Outlook}
This paper brings together what is known about the Gregory--Laflamme instability as well as providing a direct elementary mathematically rigorous proof of its existence without the use of numerics and with an explicit bound on $\mu$ and $\omega$. Note that whilst only the~$5D$ Schwarzschild black string was considered here, the result of instability readily extends to higher dimensions with the replacement of~$\omega z$ in the exponential factor with~$\sum_i\omega_iz_i$.\\

Further directions of work could be to study the non-linear problem, the extension to~$\mathrm{Kerr}_4\times\mathbb{S}^1$ or~$\mathrm{Kerr}_4\times\mathbb{R}$, the extension to charged black branes of the work~\cite{GL5}, the extension to black rings or ultraspinning Myers--Perry black holes.
\subsection{Contextual Remarks}
\subsubsection{Motivation for the Study of Higher Dimensions}
The study of higher dimensions merits a few words of motivation, since, from a physical standpoint only~$3+1$ are perceived classically. First, from a purely mathematical perspective, it is of interest to see how general relativity differs in higher dimensions from the~$4D$ case. This throws light on how general Lorentzian manifolds obeying the vacuum Einstein equation~(\ref{VE}) behave. Secondly, the physics community is very interested in higher-dimensional gravity from the point of view of string theory. Understanding how general relativity behaves in higher dimensions is therefore of relevance to the low energy limit of string theory~\cite{H1}. 
\subsubsection{Some Differences in Higher Dimensions}
In higher dimensions, many results from~$4D$ general relativity no longer hold. As shown by Hawking, in~$4D$ the cross-sections of the event horizon of an asymptotically flat stationary black hole spacetime must be topologically~$\mathbb{S}^2$ (under the dominant energy condition)~\cite{HawkingEllis}. In higher dimensions, it is possible to construct explicit examples of black hole spacetimes with non-spherical cross-sectional horizon topology. For example, the black ring solution with horizon topology~$\mathbb{S}^2\times \mathbb{S}^1$~\cite{Harvey}. In higher dimensions, there also exists a generalized Kerr solution known as the Myers--Perry black hole~\cite{MP}, which has cross-sectional horizon topology~$\mathbb{S}^3$. Hawking's theorem has been generalized to higher dimensions~\cite{Galloway}, which shows that the horizon topology must be of positive scalar curvature. In 5D under the assumptions of stationarity, asymptotic flatness, two commuting axisymmetries and `rod structure' black holes are unique and further the horizon topology is either~$\mathbb{S}^3$,~$\mathbb{S}^1\times\mathbb{S}^2$ or lens space~\cite{Hollands}. \\

In~$4D$ it is conjectured that maximal developments of `generic' asymptotically flat initial data sets can asymptotically be described by a finite number of Kerr black holes. This `final state conjecture' cannot generalize immediately since there exist at least two distinct families of black hole solutions that can have the same mass and angular momentum: the Myers--Perry black hole and the black ring. Moreover, there exist distinct black ring solutions with the same mass and angular momentum~\cite{H1,H2}. The final state conjecture may need to be modified to include the property of stability.

\subsubsection{Related Works}
A few other works are of relevance to this discussion. The review paper~\cite{H1} and book chapter~\cite{H2} discuss the black ring solution~\cite{Harvey} in great detail.  This relates to the work presented here since the Gregory--Laflamme instability is often heuristically invoked when discussing higher-dimensional black hole solutions. In particular, if the black ring of study has a large radius and is sufficiently thin then it `looks like' a Schwarzschild black string and therefore would be susceptible to the Gregory--Laflamme instability. There has been heuristic and numerical results to give evidence to this claim~\cite{GL4,Jorge}. Finally, note that in 2018 \citep{GBenomio} produced the first mathematically rigorous result on the stability problem for the black ring spacetime.
\section*{Acknowlegdements}
First and foremost, I would like to express my gratitude to my supervisor Professor Mihalis Dafermos for introducing me to this project, for his guidance and his comments on this manuscript. In particular, thanks also go to Harvey Reall for his comments on this manuscript. Thanks also go to Claude Warnick, Pierre Raphael, Yakov Shlapentokh-Rothman, Christoph Kehle, Rita Teixeira da Costa and Renato Velozo-Ruiz for many useful discussions. Lastly, I'd like to thank Kasia for her support. 
\pagebreak
\section{Linear Perturbation Theory}\label{LPT}
This section provides a derivation and review of the linearised vacuum Einstein equation~(\ref{PertE}) around a general spacetime background metric~$(M,g)$ satisfying the vacuum Einstein equation~(\ref{VE}).
\subsection{Linearised Vacuum Einstein Equation}
Consider a Lorentzian manifold~$(M,g)$ with metric satisfying the vacuum Einstein equation 
\begin{eqnarray}
\mathrm{Ric}_{g}=0.
\end{eqnarray}
In this section a `perturbation' of the spacetime metric will be discussed. This will be represented by a new metric of the form~$g+\epsilon h$ with~$\epsilon>0$.~$h$ here is a symmetric bilinear form on the fibres of~$TM$. In the following, a series of results on how various quantities change to~$O(\epsilon)$ (the linear level) are derived. This will result in an expression for the Ricci tensor under such a perturbation to linear order.
\begin{remark}
An important point to note that indices are raised and lowered here with respect to~$g$.
\end{remark}

\noindent 
\begin{prop}[Change in the Ricci Tensor]\label{Ricci}
Consider a Lorentzian manifold~$(M,g)$. Suppose the metric~$\tilde{g}_{ab}=g_{ab}+\epsilon h_{ab}$ is a Lorentzian metric. Then the Ricci tensor,~$(\widetilde{\mathrm{Ric}_g})_{ab}$, of~$\tilde{g}_{ab}$ to~$O(\epsilon)$ is
\begin{eqnarray}
(\widetilde{\mathrm{Ric}_g})_{ab}= ({\mathrm{Ric}_g})_{ab}-\epsilon\frac{1}{2}\Delta_Lh_{ab},
\end{eqnarray}
where~$\Delta_L$ denotes the Lichnerowicz operator given by
\begin{eqnarray}
\Delta_Lh_{ab}=\Box_g h_{ab}+2{{{R_{a}}^{c}}_b}^dh_{cd}-2({\mathrm{Ric}_g})_{c(a}{h_{b)}}^{c}-2\nabla_{(a}\nabla^{c}h_{b)c}+\nabla_a\nabla_bh,
\end{eqnarray}
and~$h=g^{ab}h_{ab}$.
\end{prop}
\begin{proof}
Direct computation. 
\end{proof}
If one assumes~$g$ satisfies the vacuum Einstein equation~(\ref{VE}) and~$g+\epsilon h$ satisfies the vacuum Einstein equation~(\ref{VE}) to~$O(\epsilon)$ then it follows from proposition \ref{Ricci} that~$h$ must satisfy
\begin{eqnarray}
\Box_gh_{ab}+\nabla_a\nabla_bh-2\nabla_{(b}\nabla^{c}h_{a)c}+2{{{R_{a}}^{c}}_b}^dh_{cd}=0\label{LVEE}
\end{eqnarray}
to~$O(\epsilon)$. In what follows, equation~(\ref{LVEE}) will be called the linearised vacuum Einstein equation.
This will be the main equation of interest, with~$g$ the Schwarzschild black string metric
\begin{eqnarray}
g:=-D(r)dt\otimes dt+\frac{1}{D(r)}dr\otimes dr+r^2\Big(d\theta\otimes d\theta+\sin^2\theta d\phi\otimes d\phi\Big)+dz\otimes dz,\quad D(r)=1-\frac{2M}{r}.
\end{eqnarray}
\subsection{Pure Gauge Solutions in Linearised Theory}\label{GLT}
The vacuum Einstein equation~(\ref{VE}) is a system of second order quasilinear partial differential equations of the pair~$(M,g)$ which are invariant under the diffeomorphisms of~$M$. This means that for given initial data, the vacuum Einstein equation~(\ref{VE}) only determines a spacetime unique up to diffeomorphism, i.e., if there exists a diffeomorphism~$\Phi:M\rightarrow M$ then~$(M,g)$ and~$(M,\Phi_*(g))$ are equivalent solutions of the vacuum Einstein equation~(\ref{VE}). For constructing spacetimes, one often imposes conditions on local coordinates called a gauge choice. For linearised theory this can be formulated as follows. \\

Consider a Lorentzian manifold~$(M,\tilde{g}:=g+\epsilon h)$ with~$\epsilon>0$. Let~$\{\Phi_{\tau}\}$ be a~$1$-parameter family of diffeomorphisms generated by a vector field~$X$ and define~$\xi:=\tau X\in TM$. Then from the definition of the Lie derivative one has
\begin{eqnarray}
(\Phi_{\tau})_*(\tilde{g})=\tilde{g}+\mathcal{L}_{\xi}g+\mathcal{O}(\epsilon^2)
\end{eqnarray}
if one treats~$\tau=\mathcal{O}(\epsilon)$. So in the context of linearised theory, one considers two solutions to the linearised vacuum Einstein equation~(\ref{LVEE}),~$h_1$ and~$h_2$, as equivalent if
\begin{eqnarray}
h_2=h_1+\mathcal{L}_{\xi}g\Longleftrightarrow (h_2)_{ab}=(h_1)_{ab}+2\nabla_{(a}\xi_{b)}
\end{eqnarray}
for some vector field~$\xi\in TM$. 
\begin{definition}[Pure Gauge Solution]
Let $(M,g)$ be a vacuum spacetime. A solution~$h$ to the linearised vacuum Einstein equation~(\ref{LVEE}) will be called pure gauge if there exists a vector field~$\xi\in TM$ such that 
\begin{eqnarray}
h_{ab}=2\nabla_{(a}\xi_{b)}.
\end{eqnarray}
The notation~$h_{\mathrm{pg}}$ will be used to denote a pure gauge solution to the linearised vacuum Einstein equation~(\ref{LVEE}).
\end{definition}
Showing that a solution~$h$ to the linearised vacuum Einstein equation~(\ref{LVEE}) is \emph{not} pure gauge is tantamount to showing that~$h$ is not equivalent to the trivial solution. It is thus essential that the solution constructed in this paper \emph{not} be pure gauge. 
\pagebreak

\section{Analysis in Spherical Gauge}\label{spherical}
In this section a mode solution,~$h$, of the linearised vacuum Einstein equation~(\ref{LVEE}) on the exterior~$\mathcal{E}_A$ of the Schwarzschild black string spacetime~$\mathrm{Sch}_4\times \mathbb{R}$ or~$\mathrm{Sch}_4\times \mathbb{S}^1_R$ is considered. One makes the additional assumption that this mode solution preserves the spherical symmetry of~$\mathrm{Sch}_4$. So in particular the solution can be expressed in $(t,r,\theta,\phi,z)$ coordinates as
\begin{eqnarray}
h_{\alpha\beta}=e^{\mu t+i\omega z}\begin{pmatrix}
H_{tt}(r)&H_{tr}(r)&0&0&H_{tz}(r)\\
H_{tr}(r)&H_{rr}(r)&0&0&H_{rz}(r)\\
0&0&H_{\theta\theta}(r)&0&0\\
0&0&0&H_{\theta\theta}(r)\sin^2\theta&0\\
H_{tz}(r)&H_{rz}(r)&0&0&H_{zz}(r)
\end{pmatrix} \label{GP}
\end{eqnarray}
where $\alpha,\beta\in\{t,r,\theta,\phi,z\}$. Moreover, in search of instability, the most interesting case for the present work is~$\mu>0$. \\

This section contains the analysis of the ODEs resulting from the linearised Einstein vacuum equation~(\ref{LVEE}) for a mode solution of the form~(\ref{GP}) when it is expressed in spherical gauge. 
\begin{definition}[Spherical Gauge]
A mode solution~$h$ of the linearised vacuum Einstein equation~(\ref{LVEE}) on the exterior~$\mathcal{E}_A$ of the Schwarzschild black string spacetime~$\mathrm{Sch}_4\times\mathbb{R}$ is said to be in spherical gauge if it is of the form
\begin{eqnarray}
h_{\mu\nu}=e^{\mu t+i\omega z}\begin{pmatrix}
H_t(r)&\mu H_v(r)&0&0&0\\
\mu H_v(r)&H_r(r)&0&0&-i\omega H_v(r)\\
0&0&0&0&0\\
0&0&0&0&0\\
0&-i\omega H_v(r)&0&0&H_z(r)
\end{pmatrix}\label{gaugechoice}.
\end{eqnarray}
For the Schwarzschild black string spacetime~$\mathrm{Sch}_4\times\mathbb{S}^1_R$ one makes the same definition with the additional assumption that~$\omega R\in \mathbb{Z}$.
\end{definition}
\begin{remark}
The terminology `spherical gauge' is motivated by the fact that a mode solution of this form preserves the area of the spheres of the original spacetime. 
\end{remark}
First, it is shown in section \ref{Consistency} that one can impose the gauge consistently at the level of modes, i.e., if there is a mode solution of the form~(\ref{GP}), with~$\mu\neq 0$ and either~$\omega\neq 0$ or~$\frac{dH_{tz}}{dr}-H_{rz}=0$, then there is a mode
solution of the form~(\ref{gaugechoice}) differing from the original one
by a pure gauge solution. In the case where~$H_{tz}=0$,~$H_{rz}=0$ and~$H_{zz}=0$ this consistency condition is already implicit in~\cite{Wiseman,GL4}. In section \ref{HMW}, the original decoupling of the ODEs resulting from the linearised vacuum Einstein equation~(\ref{LVEE}) and the spherical gauge ansatz~(\ref{gaugechoice}) is reproduced from~\cite{GL4}. This decoupling results in a single ODE for the component~$H_z(r)$ in equation~(\ref{gaugechoice}). It is then shown, in section \ref{ExcludePG}, that if~$\omega\neq0$, then mode solutions in spherical gauge~(\ref{gaugechoice}) are not pure gauge. Next, in section \ref{RC}, the admissible boundary conditions for the solution at the future event horizon~$\mathcal{H}_A^+$ and finiteness conditions at spacelike infinity~$i_A^0$ are identified. Note this issue is subtle since, in general, \underline{both} `basis' elements for a mode solution~$h$ of the form~(\ref{gaugechoice}) are, in fact, singular at the future event horizon~$\mathcal{H}_A^+$ in this gauge. By adding a pure gauge perturbation, the admissible boundary conditions for the solution~$h$ in the form~(\ref{gaugechoice}) can be identified. Moreover, this pure gauge solution can be chosen such that, after adding it, the harmonic/transverse-traceless gauge~(\ref{HG}) conditions are satisfied.  Finally, in section \ref{Reduce}, the problem of constructing a linear mode instability of the form~(\ref{GP}) is reduced to showing there exists a solution to the decoupled ODE for~$H_z(r)$, with~$\mu>0$ and~$\omega\neq 0$, that satisfies the admissible boundary conditions at the future event horizon~$\mathcal{H}^+_A$ and spacelike infinity~$i_A^0$ (see proposition~\ref{Asymp}). 
\subsection{Consistency}\label{Consistency}
In the paper~\cite{GL4}, it is stated that any mode solution of the form in equation~(\ref{GP}) with~$H_{tz}=0$,~$H_{rz}=0$ and~$H_{zz}=0$ can be brought to the spherical gauge form~(\ref{gaugechoice}) by the addition of a pure gauge solution. Slightly more generally, one, in fact, has the following:
\begin{prop}[Consistency of the Spherical Gauge]\label{SphCon}Consider a mode solution~$h$ to the linearised Einstein vacuum equation~(\ref{LVEE}) on the exterior~$\mathcal{E}_A$ of the Schwarzschild black string spacetime~$\mathrm{Sch}_4\times \mathbb{R}$ or~$\mathrm{Sch}_4\times \mathbb{S}^1_R$ of the form~(\ref{GP}) with~$\mu\neq 0$. Further suppose that either~$\omega\neq 0$ or~$\frac{d}{dr}H_{tz}-\mu H_{rz}=0$.  Then there exists a pure gauge solution~$h_{\mathrm{pg}}$ such that~$h+h_{\mathrm{pg}}$ is of the form~(\ref{gaugechoice}). It is in this sense that the spherical gauge~(\ref{gaugechoice}) can be consistently imposed on the exterior~$\mathcal{E}_A$ of the Schwarzschild black string~$\mathrm{Sch}_4\times \mathbb{R}$ or~$\mathrm{Sch}_4\times \mathbb{S}^1_R$.
\end{prop}
\begin{proof}
From section \ref{GLT}, a pure gauge solution is given by~$h_{\mathrm{pg}}=2\nabla_{(a}\xi_{b)}$ for a vector field~$\xi$. So,~$\tilde{h}_{ab}=h_{ab}+2\nabla_{(a}\xi_{b)}$ is the new mode solution. Consider a diffeomorphism generating vector field of the form~$ \xi=e^{\mu t+i\omega z}(\zeta_t(r),\zeta_r(r),0,0,\zeta_z(r))$. \\

If~$\omega\neq 0$, one can take
\begin{align}
\begin{split}
\zeta_t(r)&=\frac{i r(r-2M)}{2M\omega}\big(\partial_r H_{tz}(r)-\mu H_{rz}(r)\big)+\frac{r(r-2M)}{2M}H_{tr}(r)-\frac{r\mu }{2M}H_{\theta\theta}(r),\\
\zeta_r(r)&=-\frac{H_{\theta\theta}(r)}{2(r-2M)},\\
\zeta_z(r)&=-\frac{\big(H_{tz}(r)+i\omega\zeta_t(r)\big)}{\mu}
\end{split}
\end{align}
and immediately verify that~$\tilde{h}$ is of the form~(\ref{gaugechoice}). \\

If~$\frac{d}{dr}H_{tz}-\mu H_{rz}=0$, then one can take
\begin{eqnarray}
\zeta_t(r)=\frac{r(r-2M)}{2M}H_{tr}(r)-\frac{r\mu }{2M}H_{\theta\theta}(r),\quad
\zeta_r(r)=-\frac{H_{\theta\theta}(r)}{2(r-2M)},\quad
\zeta_z(r)=-\frac{\big(H_{tz}(r)+i\omega\zeta_t(r)\big)}{\mu}
\end{eqnarray}
and immediately verify that~$\tilde{h}$ is of the form~(\ref{gaugechoice}).
\end{proof}

\subsection{Reduction to ODE}\label{HMW}
Under a spherical gauge ansatz~(\ref{gaugechoice}) with~$\mu\neq 0$ and~$\omega\neq 0$, the linearised vacuum Einstein equation~(\ref{LVEE}) reduces to a system of coupled ODEs for the components~$H_t$,~$H_v$,~$H_r$ and~$H_z$. This system of ODEs can be decoupled to the single ODE for~$\mathfrak{h}:=H_z$
\begin{eqnarray}
\frac{d^2\mathfrak{h}}{dr^2}(r)+P_{\omega}(r)\frac{d\mathfrak{h}}{dr}(r)+\Big(Q_{\omega}(r)-\frac{\mu^2r^2}{(r-2M)^2}\Big)\mathfrak{h}(r)=0,\label{HZ}
\end{eqnarray}
with 
\begin{eqnarray}
P_{\omega}(r)&:=&\frac{12M}{r(\omega^2r^3+2M)}-\frac{5}{r}+\frac{1}{r-2M}\label{P},\\
Q_{\omega}(r)&:=&\frac{6M}{r^2(r-2M)}-\frac{r\omega^2}{r-2M}-\frac{12M^2}{r^2(r-2M)(\omega^2r^3+2M)}\label{Q}.
\end{eqnarray}
The following proposition establishes this decoupling of the linearised vacuum Einstein equation~(\ref{LVEE}) to the ODE~(\ref{HZ}) and the construction of a mode solution~$h$ in spherical gauge~(\ref{gaugechoice}) from a solution~$\mathfrak{h}$ to the ODE~(\ref{HZ}).
\begin{prop}\label{construct}
Given a mode solution~$h$ in spherical gauge~(\ref{gaugechoice}) with~$\mu\neq0$ and~$\omega\neq 0$ on the exterior~$\mathcal{E}_A$ of the Schwarzschild black string~$\mathrm{Sch}_4\times\mathbb{R}$ or~$\mathrm{Sch}_4\times\mathbb{S}^1_R$, the ODE~(\ref{HZ}) is satisfied by~$h_{zz}$. Conversely, given a~$C^2((2M,\infty))$ solution~$\mathfrak{h}(r)$ to the ODE~(\ref{HZ}) with~$\omega\neq 0$ and~$\mu\neq 0$, one can construct a mode solution~$h$ in spherical gauge~(\ref{gaugechoice}) to the linearised vacuuum Einstein equation~(\ref{LVEE}) on the exterior~$\mathcal{E}_A$ of the Schwarzschild black string~$\mathrm{Sch}_4\times\mathbb{R}$. If~$\omega R\in \mathbb{Z}$ then~$h$ induces a mode solution on~$\mathrm{Sch}_4\times \mathbb{S}^1_R$. 
\end{prop}
\begin{proof}
Let~$h$ be a mode solution in spherical gauge~(\ref{gaugechoice}) with~$\mu\in \mathbb{R}$ and~$\omega\in \mathbb{R}$ satisfying the linearised vacuum Einstein equation~(\ref{LVEE}) on the exterior~$\mathcal{E}_A$ of the Schwarzschild black string~$\mathrm{Sch}_4\times \mathbb{R}$ or~$\mathrm{Sch}_4\times \mathbb{S}^1_R$. Equivalently, the following system of ODE has to be satisfied:
\begin{eqnarray}
\mu\omega H_r&=&\frac{2M\mu\omega}{r(r-2M)}H_v,\label{eqs1}\\
\mu \omega^2H_v&=&\frac{ \mu }{2 }\frac{dH_z}{dr}-\frac{\mu (r-2 M)H_r  }{ r^2 }-\frac{ \mu M   H_z  }{2 r (r - 2 M) },\label{eqs2}\\
 \omega \frac{d H_t}{dr}&=&\frac{\omega M H_t }{
  r( r - 2M)} - \frac{\omega(r-2 M  ) (2 r-3 M ) H_r }{r^3} + 
 2 \mu^2 \omega H_v, \label{eqs3}\\
 H_t&=&\frac {(r - 2 M)  ( r ( \omega^2-\mu^2 )-2 M \omega^2 )  
        } { M }H_v  +\frac {(r - 2 M)^2(r+M)  } { M r^2}H_r \label{eqs4}\\
           &&+\frac {(r - 2 M)^3 } {2 M r}\frac{dH_r}{dr} -\frac {(r - 2 M)^2 } {2 M} \frac{dH_z }{dr}+\frac {r(r - 2 M)} {2 M }\frac{dH_t}{dr}, \nonumber\\
 \frac{d^2H_z}{dr^2}&=&\omega^2H_r-\frac{\omega^2r^2}{(r-2M)^2}H_t+\frac{(r-M)}{r(r-2M)}(4\omega^2H_v-2\frac{dH_z}{dr})+\frac{\mu^2r^2}{(r-2M)^2}H_z+2\omega^2\frac{dH_v}{dr},\label{eqs5}\\
    \frac{d^2H_z}{dr^2}&=&\frac{2M(2r-3M)}{r(r-2M)^3}H_t-\frac{\big(6M^2-(\mu^2+\omega^2)r^4+2Mr(\omega^2r^2-2)\big)}{r^3(r-2M)}H_r\label{eqs6}\\
    &&-\frac{2M(2M\omega^2+r(\mu^2-\omega^2))}{r(r-2M)^2}H_v+\frac{2r-3M}{r^2}\frac{dH_r}{dr}-\frac{2\mu^2 r+4M\omega^2-2\omega^2 r}{r-2M}\frac{dH_v}{dr}\nonumber\\
    &&-\frac{M}{r(r-2M)}\frac{dH_z}{dr}-\frac{M}{(r-2M)^2}\frac{dH_t}{dr}+\frac{r}{r-2M}\frac{d^2H_t}{dr^2},\nonumber
    \end{eqnarray}
    \begin{eqnarray}
    \frac{d^2H_t}{dr^2}&=&\frac{\omega^2r^4-2M\omega^2r^3-2M^2}{r^2(r-2M)^2}H_t-\Big(\mu^2+\frac{2M^2}{r^4}\Big)H_r-\frac{r\mu^2}{r-2M}H_z\label{eqs7}\\
    &&+\frac{4\mu^2r^2+4M^2\omega^2-2Mr(3\mu^2+\omega^2)}{r^2(r-2M)}H_v-\frac{2r-5M}{r(r-2M)}\frac{dH_t}{dr}-\frac{M(r-2M)}{r^3}\frac{dH_r}{dr}\nonumber\\
    &&+2\mu^2\frac{dH_v}{dr}+\frac{M}{r^2}\frac{dH_z}{dr}\nonumber.
\end{eqnarray}

Now, if~$\mu\neq 0$ and~$\omega\neq 0$, then from equations~(\ref{eqs1}) and~(\ref{eqs2}) one can find~$H_v$ in terms of~$H_z$ and~$\frac{dH_z}{dr}$. This can then be used in equation~(\ref{eqs3}) to give and equation for~$\frac{dH_t}{dr}$ in terms of~$H_t$,~$H_z$ and~$\frac{dH_z}{dr}$. All of these expressions can be used to express~$H_t$ in terms of~$H_z$,~$\frac{dH_z}{dr}$ and~$\frac{d^2H_z}{dr^2}$ via equation~(\ref{eqs4}). The resulting equations are
\begin{eqnarray}
H_r(r)&=&-\frac{M^2r}{(r-2M)^2(\omega^2r^2+2M)}H_z(r)+\frac{Mr^2}{(r-2M)(\omega^2r^2+2M)}\frac{dH_z}{dr},\label{HR}\\
H_v(r)&=&-\frac{Mr^2}{(2(r-2M)(\omega^2r^2+2M)}H_z(r)+\frac{r^3}{2(\omega^2r^2+2M)}\frac{dH_z}{dr},\label{HV}\\
H_t(r)&=&\frac{2M^2(r-3M)+M\omega^2r^3(2r-5M)-\omega^4r^6(r-2M)}{r(\omega^2r^3+2M)^2}H_z,\label{HT}\\
&&-\frac{2(r-2M)(M(r-4M)+(2r-5M)\omega^2r^3)}{(\omega^2r^3+2M)^2}\frac{dH_z}{dr}+\frac{r(r-2M)^2}{\omega^2r^3+2M}\frac{d^2H_z}{dr^2}\nonumber.
\end{eqnarray}
Finally, one can use the above expressions to obtain a decoupled ODE for~$\mathfrak{h}:=H_z$, namely
\begin{eqnarray}
\frac{d^2\mathfrak{h}}{dr^2}(r)+P_{\omega}(r)\frac{d\mathfrak{h}}{dr}(r)+\Big(Q_{\omega}(r)-\frac{\mu^2r^2}{(r-2M)^2}\Big)\mathfrak{h}(r)=0,
\end{eqnarray}
with 
\begin{eqnarray}
P_{\omega}(r)&:=&\frac{12M}{r(\omega^2r^3+2M)}-\frac{5}{r}+\frac{1}{r-2M},\\
Q_{\omega}(r)&:=&\frac{6M}{r^2(r-2M)}-\frac{r\omega^2}{r-2M}-\frac{12M^2}{r^2(r-2M)(\omega^2r^3+2M)}.
\end{eqnarray}

Conversely, given any~$C^2((2M,\infty))$ solution~$\mathfrak{h}(r)$ to the ODE~(\ref{HZ}) with~$\omega\neq 0$ and~$\mu\neq 0$ one can define~$H_z(r)=\mathfrak{h}(r)$. Since~$\omega\neq 0$, one can use equations~(\ref{HR})--(\ref{HT}) to construct~$H_{t}(r)$,~$H_{r}(r)$ and~$H_{v}(r)$. These then define the components of a mode solution~$h$ in spherical gauge~(\ref{gaugechoice}). Explicitly
\begin{eqnarray}
h=e^{\mu t+i\omega z}\begin{pmatrix}
H_t(r)&\mu H_v(r)&0&0&0\\
\mu H_v(r)&H_r(r)&0&0&-i\omega H_v(r)\\
0&0&0&0&0\\
0&0&0&0&0\\
0&-i\omega H_v(r)&0&0&H_z(r)
\end{pmatrix}.
\end{eqnarray}
If the ODE~(\ref{HZ}) is satisfied and~(\ref{HR})--(\ref{HT}) define~$H_r$,~$H_v$ and~$H_t$, then equations~(\ref{eqs1})--(\ref{eqs7}) are also satisfied. Therefore, a mode solution~$h$ constructed in this manner solves the linearised vacuum Einstein equation~(\ref{LVEE}) on the exterior~$\mathcal{E}_A$ of the Schwarzschild black string~$\mathrm{Sch}_4\times \mathbb{R}$. If~$\omega R\in \mathbb{Z}$ then this construction also gives a mode solution~$h$ which solves the linearised vacuuum Einstein equation~(\ref{LVEE}) on the exterior~$\mathcal{E}_A$ of the Schwarzschild black string~$\mathrm{Sch}_4\times \mathbb{S}^1_R$.
\end{proof}
\begin{remark}
If~$\omega=0$ and~$\mu\neq 0$, then one can add an additional pure gauge solution~$h_{\mathrm{pg}}$ to a mode solution~$h$ in spherical gauge~(\ref{gaugechoice}) such that~$h+h_{\mathrm{pg}}$ is also in spherical gauge~(\ref{gaugechoice}) with~$H_t(r)\equiv 0$. The relevant choice of pure gauge solution is given by~$(h_{\mathrm{pg}})_{ab}=2\nabla_{(a}\xi_{b)}$ with
\begin{eqnarray}
\xi=e^{\mu t}\Big(-\frac{H_t(r)}{2\mu},0,0,0,0\Big).
\end{eqnarray}

A mode solution~$h$ in spherical gauge with~$H_t(r)\equiv 0$ satisfying the linearised vacuum Einstein equation~(\ref{LVEE}) on the exterior~$\mathcal{E}_A$ of the Schwarzschild black string is then again equivalent to the system of ODE~(\ref{eqs1})-(\ref{eqs7}) (with~$\omega=0$ and~$H_t\equiv 0$) being satisfied. Equations~(\ref{eqs1}) and~(\ref{eqs3}) are automatically satisfied by $\omega=0$. The equation~(\ref{eqs5}) automatically gives the decoupled equation~(\ref{HZ}) for $H_z$. Then, equation~(\ref{eqs2}) can be solved for~$H_r$ in terms of~$H_z$ and~$\frac{dH_z}{dr}$. Equation~(\ref{eqs4}) can be used to solve for~$H_v$ in terms of~$H_z$ and~$\frac{dH_z}{dr}$. At this point, the equations~(\ref{eqs6}) and~(\ref{eqs7}) are automatically satisfied. Therefore, again a solution to the ODE~(\ref{HZ}) induces a mode solution in spherical gauge with $H_t=0$.
\end{remark}
\subsection{Excluding Pure Gauge Perturbations}\label{ExcludePG}
\noindent This section contains a proof that if~$\omega\neq 0$ then a non-trivial mode solution~$h$ of the form~(\ref{gaugechoice}) cannot be a pure gauge solution. More precisely, one has the following proposition: 
\begin{prop}\label{notpuregaugeS}
Suppose~$\omega\neq 0$ and~$\mu\in \mathbb{R}$. A non-trivial mode solution~$h$ in spherical gauge~(\ref{gaugechoice}) of the linearised vacuum Einstein equation~(\ref{LVEE}) on the exterior~$\mathcal{E}_A$ of the Schwarzschild black string~$\mathrm{Sch}_4\times \mathbb{R}$ or~$\mathrm{Sch}_4\times\mathbb{S}^1_R$ cannot be pure gauge. 
\end{prop}
\begin{proof}
Assume~$\omega\neq 0$. If $h$ is pure gauge, it must be possible to write~$h_{ab}=2\nabla_{(a}\xi_{b)}$ for some vector field~$\xi$. Therefore one finds
\begin{eqnarray}
h_{zz}&=&H_z(r)e^{\mu t+i\omega z} \implies2\partial_z\xi_z=H_z(r)e^{\mu t+i\omega z}\label{geq1},\\
h_{z\theta}&=&0\implies\partial_{\theta}\xi_z+\partial_z\xi_{\theta}=0\label{geq2}.
\end{eqnarray}
Applying~$\partial_z$ to the equation~(\ref{geq2}), using that partial derivatives commute and that, from equation~(\ref{geq1}),~$\partial_z\xi_z$ clearly does not depend on~$\theta$ gives 
\begin{eqnarray}
\partial_z^2\xi_{\theta}=0.
\end{eqnarray}

Next,~$h_{\theta\theta}=0$ implies
\begin{eqnarray}
\partial_{\theta}\xi_{\theta}-\Gamma_{\theta\theta}^r\xi_r=0\label{geq4}.
\end{eqnarray}
From appendix \ref{ChR},~$\Gamma_{\theta\theta}^r=(r-2M)$. Hence, taking two derivatives of~(\ref{geq4}) in the~$z$ direction and using~$\partial_z^2\xi_{\theta}=0$ gives
\begin{eqnarray}
\Gamma_{\theta\theta}^r\partial_z^2\xi_r=(r-2M)\partial_z^2\xi_r=0.
\end{eqnarray}
Therefore,~$\partial_z^2\xi_r=0$ on~$\mathcal{E}_A$.\\

From the~$h_{rr}$ component one has, 
\begin{eqnarray}
2\partial_r\xi_r-2\Gamma_{rr}^r\xi_r=2\partial_r\xi_r+\frac{2M}{r(r-2M)}\xi_r=H_re^{\mu t+i\omega z}\label{geq3}
\end{eqnarray}
where one uses~$\Gamma^r_{rr}=-\frac{M}{r(r-2M)}$ from appendix \ref{ChR}. Taking the second~$z$ derivative of equation~(\ref{geq3}) and using~$\partial_z^2\xi_r=0$ on~$\mathcal{E}_A$ gives
\begin{eqnarray}
\omega^2H_r=0\quad \text{on}\quad \mathcal{E}_A.
\end{eqnarray}
Since~$\omega\neq 0$, this implies~$H_r\equiv 0$ on the exterior~$\mathcal{E}_A$. Since~$\omega\neq 0$, equation~(\ref{eqs1}) implies that if~$H_r=0$ on~$\mathcal{E}_A$, then~$H_v\equiv 0$ on~$\mathcal{E}_A$. Using the~$h_{zr}$ component, one finds
\begin{eqnarray}
\partial_z\xi_r+\partial_r\xi_z=-i\omega H_ve^{\mu t+i\omega z}=0\implies \partial_r(\partial_z\xi_z)=0\implies \frac{dH_z}{dr}=0\quad \text{on}\quad \mathcal{E}_A,
\end{eqnarray}
where one uses the identity~$\partial_z^2\xi_r=0$ on~$\mathcal{E}_A$ in the first implication and that~$\partial_z\xi_z=H_z(r)e^{\mu t+i\omega z}$ in the second implication. The linearised vacuum Einstein equation~(\ref{LVEE}) under this ansatz (equation~(\ref{eqs2})) then implies~$H_z\equiv 0$ on~$\mathcal{E}_A$ and therefore, from equations~(\ref{eqs3}) and~(\ref{eqs4}),~$H_t\equiv 0$ on~$\mathcal{E}_A$. Hence,~$h\equiv 0$ on~$\mathcal{E}_A$. 
\end{proof}
\subsection{Admissible Boundary Conditions}\label{RC}
One can construct two sets of distinguished solutions to the ODE~(\ref{HZ}) associated to the ``end points" of the interval~$(2M,\infty)$. Note that, by definition~\ref{RegSingDef} from appendix~\ref{Sing},~$r=2M$ is a regular singularity, as $2M$ is not an ordinary point and
\begin{eqnarray}
(r-2M)P_{\omega}(r)\quad \text{ and }\quad (r-2M)^2\Big(Q_{\omega}(r)-\frac{\mu^2r^2}{(r-2M)^2}\Big)
\end{eqnarray}
are analytic near~$r=2M$. By definition~\ref{irregsinginf}, the ODE~(\ref{HZ}) has an irregular singularity at infinity, since there exist convergent series expansions
\begin{eqnarray}
P_{\omega}(r)=\sum_{n=0}^{\infty}\frac{p_n}{z^n}\quad \text{ and }\quad Q_{\omega}(r)-\frac{\mu^2r^2}{(r-2M)^2}=\sum_{n=0}^{\infty}\frac{q_n}{z^n}
\end{eqnarray}
in a neighbourhood of infinity with~$p_0=0$,~$p_1=-4$,~$q_0=-(\omega^2+\mu^2)$ and~$q_1=-2M(\omega^2+2\mu^2)$. The asymptotic analysis of the ODEs around these points is examined in the following two subsections. This analysis of the ODE~(\ref{HZ}) near~$r=2M$ and~$r=\infty$ will lead to the identification of the admissible boundary conditions for a mode solution~$h$ in spherical gauge~(\ref{gaugechoice}) of the linearised Einstein vacuum equation~(\ref{LVEE}).  
\subsubsection{The Future Event Horizon~$\mathcal{H}^+_A$}\label{asymhor}
The goal of this section is to identify the admissible boundary conditions for a solution~$\mathfrak{h}$ to the ODE~(\ref{HZ}) near~$r=2M$. This requires one to understand the behaviour near~$r=2M$ of the mode solution~$h$ in spherical gauge~(\ref{gaugechoice}) of the linearised vacuum Einstein equation~(\ref{LVEE}) which results (through the construction in proposition~\ref{construct}) from~$\mathfrak{h}$.\\

Associated with the future event horizon $\mathcal{H}^+_A$, there exists a basis~$\mathfrak{h}^{2M,\pm}$ for solutions to the ODE~(\ref{HZ}). From~$\mathfrak{h}^{2M,\pm}$ one can examine the behavior near~$r=2M$ of any mode solution~$h$ in spherical gauge~(\ref{gaugechoice}) with~$\mu\neq 0$ and~$\omega\neq 0$ through proposition \ref{construct}. A mode solution~$h$ in spherical gauge~(\ref{gaugechoice}) with~$\mu> 0$ and~$\omega\neq 0$ constructed from~$\mathfrak{h}^{2M,-}$ never smoothly extends to the future event horizon. A mode solution~$h$ in spherical gauge~(\ref{gaugechoice}) with~$\mu>0$ and~$\omega\neq 0$ constructed from~$\mathfrak{h}^{2M,+}$ also does not smoothly extend to the future event horizon unless~$\mu$ satisfies particular conditions. However, if~$h$ is a mode solution in spherical gauge~(\ref{gaugechoice}) with~$\mu>0$ and~$\omega\neq 0$ constructed from~$\mathfrak{h}^{2M,+}$ then, after the addition of a pure gauge solution~$h_{\mathrm{pg}}$, it turns out one can smoothly extend~$h+h_{\mathrm{pg}}$ to the future event horizon. Moreover, it will be shown that~$h+h_{\mathrm{pg}}$ satisfies the harmonic/transverse-traceless gauge~(\ref{HG}) conditions. This will be the content of proposition~\ref{A1}. \\

First, some preliminaries. The coefficients of the ODE~(\ref{HZ}) extend meromorphically to~$r=2M$ and behave asymptotically as
\begin{eqnarray}
P_{\omega}(r)=\frac{1}{r-2M}+\mathcal{O}(1)\qquad Q_{\omega}(r)-\frac{\mu^2r^2}{(r-2M)^2}=-\frac{4M^2\mu^2}{(r-2M)^2}+\mathcal{O}\Big(\frac{1}{r-2M}\Big).
\end{eqnarray}
So one may write the ODE~(\ref{HZ}) as
\begin{eqnarray}
\frac{d^2\mathfrak{h}}{dr^2}+\Big(\frac{1}{r-2M}+\mathcal{O}(1)\Big)\frac{d\mathfrak{h}}{dr}-\Big(\frac{4M^2\mu^2}{(r-2M)^2}+\mathcal{O}\Big(\frac{1}{r-2M}\Big)\Big)\mathfrak{h}&=&0.\label{AsympODE}
\end{eqnarray}

From appendix~\ref{Sing}, the indicial equation associated to the ODE~(\ref{AsympODE}) is
\begin{eqnarray}
I(\alpha)=\alpha^2-4M^2\mu^2,
\end{eqnarray}
which has roots
\begin{eqnarray}
\alpha_{\pm}:=\pm 2M\mu.
\end{eqnarray}
If~$\alpha_+-\alpha_-=4M\mu \not\in \mathbb{Z}$, then one can deduce from theorem~\ref{ThE1} the asymptotic basis for solutions near~$r=2M$. If~$\alpha_+-\alpha_-=4M\mu \in \mathbb{Z}$ then the relevant result for the asymptotic basis of solutions is theorem~\ref{ThE2}. Combining the results of theorems \ref{ThE1} and \ref{ThE2} one has the following basis for solutions for~$\mu>0$
\begin{eqnarray}
\mathfrak{h}^{2M,+}(r)&:=&(r-2M)^{2M\mu}\sum_{n=0}^{\infty}a^{+}_n(r-2M)^n,\label{AHP}\\
\mathfrak{h}^{2M,-}(r)&:=&\begin{cases}
\sum_{n=0}^{\infty}a^-_{n}(r-2M)^{n-2M\mu}+C_N\mathfrak{h}^{2M,+}(r)\ln(r-2M)\quad\text{if}\quad 4M\mu=N\in\mathbb{Z}_{>0}\label{AHM}\\
(r-2M)^{-2M\mu}\sum_{n=0}^{\infty}a^{-}_n(r-2M)^{n}\qquad \text{otherwise},
\end{cases}
\end{eqnarray}
where the coefficents~$a_n^+$,~$a_n^-$ and the anomalous term~$C_N$ can be calculated recursively (see theorems \ref{ThE1} and \ref{ThE2}). A general solution to the ODE~(\ref{HZ}) will be of the form
\begin{eqnarray}
\mathfrak{h}(r)=k_1\mathfrak{h}^{2M,+}(r)+k_2\mathfrak{h}^{2M,-}(r)\label{horasymbasis}
\end{eqnarray}
with~$ k_1,k_2\in \mathbb{R}$. \\

If~$4M\mu$ is \emph{not} an integer or~$4M\mu$ is an integer and~$C_N=0$, then the asymptotic basis for solutions for~$\mu>0$ reduces to
\begin{eqnarray}
\mathfrak{h}^{2M,+}(r)&=&(r-2M)^{2M\mu}\sum_{n=0}^{\infty}a^{+}_n(r-2M)^n,\label{basis1}\\
\mathfrak{h}^{2M,-}(r)&=&(r-2M)^{-2M\mu}\sum_{n=0}^{\infty}a^{-}_n(r-2M)^{n}.\label{basis2}
\end{eqnarray}
In equations~(\ref{basis1}) and~(\ref{basis2}), the first order coefficients of the basis can be calculated to be
\begin{eqnarray}
a_1^{\pm}=\frac{\pm \mu(20M^2\omega^2-1)+4M(\mu^2-\omega^2+4M^2\mu^2\omega^2+2M^2\omega^4)}{(1\pm 4M\mu)(4M^2\omega^2+1)}.
\label{coeff1}
\end{eqnarray}

The main result of this section is the following:
\begin{prop}\label{A1}
Suppose~$\mu>0$,~$\omega\neq 0$ and let~$\mathfrak{h}$ be a solution to the ODE~(\ref{HZ}). Let~$h$ be the mode solution on the exterior~$\mathcal{E}_A$ of the Schwarzschild black string~$\mathrm{Sch}_4\times\mathbb{R}$ constructed from~$H_z=\mathfrak{h}$ in proposition~\ref{construct}. Then there exists a pure gauge solution~$h_{\mathrm{pg}}$ such that~$h+h_{\mathrm{pg}}$ extends to a smooth solution of the linearised vacuum Einstein equation~(\ref{LVEE}) at the future event horizon~$\mathcal{H}_A^+$ if~$k_2=0$, where $k_2$ is defined in equation~(\ref{horasymbasis}). Moreover,~$h+h_{\mathrm{pg}}$ can be chosen to satisfy the harmonic/transverse-traceless gauge~(\ref{HG}) conditions. 
\end{prop}
\begin{remark}
To determine admissible boundary conditions of~$\mathfrak{h}$ at $r=2M$ it is essential that one works in coordinates that extend regularly across this hypersurface. A good choice is ingoing Eddington--Finkelstein coordinates~$(v,r,\theta,\phi,z)$ defined by
\begin{eqnarray}
v=t+r_*(r),\qquad r_*(r)=r+2M\log|r-2M|.
\end{eqnarray}
Also note that for the boundary conditions to be admissible, one needs to consider all components of the mode solution~$h$ constructed from $\mathfrak{h}$ via proposition~\ref{construct}. These remarks will be implemented in the proof of proposition~\ref{A1}. 
\end{remark}
Before proving the statement of proposition~\ref{A1} it is useful to prove the following lemma:
\begin{lemma}\label{GML}
Let $h$ be a mode solution of the linearised vacuum Einstein equation~(\ref{LVEE}) of the form
\begin{eqnarray}
h_{\alpha\beta}=e^{\mu t+i\omega z}\begin{pmatrix}
H_{tt}(r)&H_{tr}(r)&0&0&0\\
H_{tr}(r)&H_{rr}(r)&0&0&0\\
0&0&H_{\theta\theta}(r)&0&0\\
0&0&0&H_{\theta\theta}(r)\sin^2\theta&0\\
0&0&0&0&0
\end{pmatrix}.\label{TTA}
\end{eqnarray}
Then $h$ satisfies the harmonic/transverse-traceless gauge conditions:
\begin{eqnarray}
\begin{cases}
\nabla^ah_{ab}=0\\
g^{ab}h_{ab}=0
\end{cases}
\end{eqnarray}
if $\omega\neq 0$.
\end{lemma}
\begin{proof}
First, it is instructive to write out explicit expressions for $\nabla_{c}h_{ab}$ and $\nabla_{c}\nabla_{d}h_{ab}$ in coordinates. These are the following:
\begin{align}
\nabla_{\gamma}h_{\alpha\beta}=&\partial_{\gamma}h_{\alpha\beta}-\Gamma^{\lambda}_{\gamma\alpha}h_{\lambda\beta}-\Gamma^{\lambda}_{\gamma\beta}h_{\alpha\lambda}\label{nab}\\
\nabla_{\gamma}\nabla_{\delta}h_{\alpha\beta}=&\partial_{\gamma}(\partial_{\delta}h_{\alpha\beta}-\Gamma^{\lambda}_{\delta\alpha}h_{\lambda\beta}-\Gamma^{\lambda}_{\delta\beta}h_{\alpha\lambda})-\Gamma^{\mu}_{\gamma\delta}(\partial_{\mu}h_{\alpha\beta}-\Gamma^{\lambda}_{\mu\alpha}h_{\lambda\beta}-\Gamma^{\lambda}_{\mu\beta}h_{\alpha\lambda})\label{nabnab}\\
&-\Gamma^{\mu}_{\gamma\alpha}(\partial_{\delta}h_{\mu\beta}-\Gamma^{\lambda}_{\delta\mu}h_{\lambda\beta}-\Gamma^{\lambda}_{\delta\beta}h_{\mu\lambda})-\Gamma^{\mu}_{\gamma\beta}(\partial_{\delta}h_{\alpha\mu}-\Gamma^{\lambda}_{\delta\alpha}h_{\lambda\mu}-\Gamma^{\lambda}_{\delta\mu}h_{\alpha\lambda})\nonumber.
\end{align}
If one takes the ansatz~(\ref{TTA}) and $\alpha=z$ in equation~(\ref{nabnab}), then, since $h_{z\beta}=0$ for all $ \beta\in \{t,r,\theta,\phi,z\}$ and, from appendix~\ref{ChR},~$\Gamma^{\lambda}_{z\beta}=0$ for all $ \beta,\lambda\in \{t,r,\theta,\phi,z\}$, 
\begin{eqnarray}
\nabla_{\gamma}\nabla_{\delta}h_{\alpha\beta}=0 \qquad (\alpha=z)\label{nabnab2}.
\end{eqnarray}
Hence, 
\begin{align}
g^{\gamma\delta}\nabla_{\gamma}\nabla_{\delta}h_{\alpha\beta}&=0 \qquad (\alpha=z)\label{nabnab3}\\
g^{\delta\beta}\nabla_{\gamma}\nabla_{\delta}h_{\alpha\beta}&=0 \qquad (\alpha=z)\label{nabnab4}.
\end{align}

Consider the linearised vacuum Einstein equation~(\ref{LVEE}) in coordinates
\begin{eqnarray}
g^{\gamma\delta}\nabla_{\gamma}\nabla_{\delta}h_{\alpha\beta}+\nabla_{\alpha}\nabla_{\beta}h-\nabla_{\alpha}\nabla^{\gamma}h_{\beta\gamma}-\nabla_{\beta}\nabla^{\gamma}h_{\alpha\gamma}+2{{{R_{\alpha}}^{\gamma}}_{\beta}}^{\delta}h_{\gamma\delta}=0.
\end{eqnarray}
Since, from equations~(\ref{nabnab3})--(\ref{nabnab4}) and, from appendix~\ref{ChR}, $R_{z\beta\gamma\delta}=0$, it follows that the linearised vacuum Einstein equation in local coordinates with $\alpha=z$ and under the ansatz~(\ref{TTA}) reduces to
\begin{eqnarray}
\nabla_{z}(\nabla_{\beta}h-\nabla^{\gamma}h_{\beta\gamma})=0\label{DH}.
\end{eqnarray}
Further, $\nabla_z=\partial_z$, so using the explicit $z$-dependence of the ansatz~(\ref{TTA}), the equation~(\ref{DH}) reduces to
\begin{eqnarray}
\omega(\nabla_{\beta}h-\nabla^{\gamma}h_{\beta\gamma})=0.
\end{eqnarray}
Since~$\omega\neq 0$, the harmonic gauge condition
\begin{eqnarray}
\nabla_{\beta}h-\nabla^{\gamma}h_{\beta\gamma}=0\label{oHG}
\end{eqnarray}
 is satisfied. If $\beta=z$ then, using equation~(\ref{nab}) and~$\nabla_z=\partial_z$, equation~(\ref{oHG}) reduces to
 \begin{eqnarray}
 \partial_zh=\omega h=0\implies  h=0\label{oTT}
 \end{eqnarray}
 since $\omega\neq 0$. Substituting~(\ref{oTT}) into equation~(\ref{oHG}) gives the transverse condition
 \begin{eqnarray}
 \nabla^{\gamma}h_{\beta\gamma}=0.
 \end{eqnarray}
\end{proof}
\begin{proof}[Proof of Proposition~\ref{A1}.]
Consider~$H_z^{2M,\pm}:=\mathfrak{h}^{2M,\pm}$ where~$\mathfrak{h}^{2M,\pm}$ are given by equations~(\ref{basis1}) and~(\ref{basis2}) with first order coefficients~(\ref{coeff1}). Taking~$k_2=0$ is equivalent to examining the basis element~$H_z^{2M,+}$. Since~$\mu>0$ and~$\omega\neq 0$, one can use proposition \ref{construct} to construct the components~$H_t$,~$H_r$ and~$H_v$ associated to~$H_z^{2M,\pm}$. Substituting the basis into equations~(\ref{HR})--(\ref{HT}), one finds
\begin{eqnarray}
H_{r}^{2M,\pm}&=&(r-2M)^{-2\pm 2M\mu}\Big(\frac{M^2( \pm 4M\mu-1)}{1+4M^2\omega^2}+\frac{M(4M^2(2\mu^2+\omega^2)\pm 6 M\mu-1)}{2(1+4M^2\omega^2)}(r-2M)\\&&+\mathcal{O}((r-2M)^2)\Big),\nonumber\\
H_{t}^{2M,\pm}&=&(r-2M)^{\pm 2M\mu}\Big(\frac{(1+ 4M\mu)( 4M\mu-1)}{4(1+4M^2\omega^2)}\\
&&+\frac{3+4M^2(8\mu^2-\omega^2)\pm2M\mu(8M^2(2\mu^2+\omega^2)-11)}{8M(1+4M^2\omega^2)}(r-2M)+\mathcal{O}((r-2M)^2)\Big),\nonumber\\
H_v^{2M,\pm}&=&(r-2M)^{-1+ 2M\mu}\Big(\frac{M^2(\pm 4M\mu-1)}{1+4M^2\omega^2}+\frac{M(2M^2(2\mu^2+\omega^2)-1\pm 5M\mu)}{1+4M^2\omega^2}(r-2M)\\
&&+\mathcal{O}((r-2M)^2)\Big)\nonumber.
\end{eqnarray}

 Consider a pure gauge solution~$h_{\mathrm{pg}}=2\nabla_{(a}\xi_{b)}$ generated by the following vector field
\begin{eqnarray}
\xi=e^{\mu t+i\omega z}\Big(-\frac{\mu H_z(r)}{2\omega^2},\frac{2\omega^2H_v(r)-\frac{dH_z}{dr}(r)}{2\omega^2},0,0,\frac{iH_z(r)}{2\omega}\Big)\label{GT}
\end{eqnarray}
where~$H_v$ is defined via equation~(\ref{HV}). This gives a new solution to the linearised vacuum Einstein equation~(\ref{LVEE})
\begin{eqnarray}
\tilde{h}_{\mu\nu}&=&h_{\mu\nu}+2\nabla_{(\mu}\xi_{\nu)}=e^{\mu t+i\omega z}\begin{pmatrix}\label{pertt}
\tilde{H}_{tt}(r)&\tilde{H}_{tr}(r)&0&0&0\\
\tilde{H}_{tr}(r)&\tilde{H}_{rr}(r)&0&0&0\\
0&0&\tilde{H}_{\theta\theta}(r)&0&0\\
0&0&0&\tilde{H}_{\theta\theta}(r)\sin^2\theta&0\\
0&0&0&0&0\\
\end{pmatrix},
\end{eqnarray}
with the following expressions for the matrix components
\begin{eqnarray}
\tilde{H}_{tt}(r)&=&c_1(r)H_z(r)+c_2(r)\frac{dH_z}{dr}(r),\label{trans1}\\
\tilde{H}_{\theta\theta}(r)&=&c_3(r)H_z(r)+c_4(r)\frac{dH_z}{dr}(r),\label{trans4}\\
\tilde{H}_{rr}(r)&=&\frac{r^2}{(r-2M)^2}\tilde{H}_{tt}(r)-\frac{2}{r(r-2M)}\tilde{H}_{\theta\theta}(r),\label{trans2}\\
\tilde{H}_{tr}(r)&=&-\frac{2M\mu}{\omega^2(2M+r^3\omega^2)}\Big(\frac{dH_z}{dr}(r)-\frac{M}{r(r-2M)}H_z(r)\Big),\label{trans3}
\end{eqnarray}
where
\begin{align}
\begin{split}
c_1(r)&:=\frac{6M^2(r-2M)}{r(\omega^2r^3+2M)^2}-\frac{2M(r-2M)}{r(\omega^2r^3+2M)}+\frac{\mu^2r^3}{\omega^2r^3+2M}-\frac{\mu^2}{\omega^2},\\
c_2(r)&:=\frac{M(r-2M)}{\omega^2r^3}-\frac{M(r-2M)}{\omega^2r^3+2M}-\frac{6M(4M^2-4Mr+r^2)}{(\omega^2r^3+2M)^2},\\
c_3(r)&:=-\frac{Mr^2}{\omega^2r^3+2M},\qquad c_4(r):=\frac{r^3(r-2M)}{\omega^2r^3+2M}-\frac{r-2M}{\omega^2}.
\end{split}
\end{align}
Note that equations~(\ref{HZ}) and~(\ref{HR})--(\ref{HT}) have been used to derive equations~(\ref{trans1})--(\ref{trans3}). By lemma~\ref{GML}, this new mode solution~(\ref{pertt}) satisfies the harmonic/transverse-traceless gauge:
\begin{eqnarray}
\begin{cases}
g^{\mu\nu}\tilde{h}_{\mu\nu}=0\label{TransTraceless}\\
\nabla^{\mu}\tilde{h}_{\mu\nu}=0.
\end{cases}
\end{eqnarray}

As remarked above, to determine admissible boundary conditions of~$\mathfrak{h}$ at $r=2M$ it is essential that one works in coordinates that extend regularly across this hypersurface. Moreover, to identify the boundary conditions to be admissible, one needs to consider all components of the mode solution~$h$ constructed from $\mathfrak{h}$ via proposition~\ref{construct}. The following formulas give the transformation to ingoing Eddington--Finkelstein coordinates for the components of the mode solution~$h$ defined in equation~(\ref{pertt}):
\begin{align}
\begin{split}
\tilde{H}_{vv}'&=\Big(\frac{\partial t}{\partial v}\Big)^2\tilde{H}_{tt},\label{TEFH}\\
\tilde{H}_{vr}'&=\Big(\frac{\partial t}{\partial v}\Big)\Big(\frac{\partial r}{\partial r}\Big)\tilde{H}_{tr}+\Big(\frac{\partial t}{\partial v}\Big)\Big(\frac{\partial t}{\partial r}\Big)\tilde{H}_{tt}\\
&=\tilde{H}_{tr}-\frac{r}{r-2M}\tilde{H}_{tt},\\
\tilde{H}_{rr}'&=\Big(\frac{\partial t}{\partial r}\Big)^2\tilde{H}_{tt}+\Big(\frac{\partial t}{\partial r}\Big)\Big(\frac{\partial r}{\partial r}\Big)\tilde{H}_{tr}+\Big(\frac{\partial r}{\partial r}\Big)^2\tilde{H}_{rr}\\
&=\frac{r^2}{(r-2M)^2}\tilde{H}_{tt}-\frac{r}{r-2M}\tilde{H}_{tr}+\tilde{H}_{rr},
\end{split}
\end{align}
where one uses~$t=v-r_*(r)$ with~$r_*(r)=r+2M\log|r-2M|$.
Explicitly, the equations~(\ref{TEFH}) can be computed to be
\begin{eqnarray}
\tilde{H}_{vv}'&=&\frac{2M(2M\mu^2r+\omega^2(\mu^2r^4-Mr+2M^2)+\omega^4r^3(r-2M))}{r(\omega^3r^3+2M\omega)^2}H_z(r)\\
&&-\frac{2M(r-2M)(\omega^2r^3(3r-7M)-2M^2)}{r^3(\omega^3r^3+2M\omega)^2}\frac{dH_z}{dr}(r),\nonumber\\
\tilde{H}_{vr}'&=&\Big(\frac{\mu(\mu r^2+M)}{r\omega^2(r-2M)}-\frac{\mu^2r^4+M\mu r^2-2M(r-2M)}{(r-2M)(\omega^2r^3+2M)}-\frac{6M^2}{(\omega^2r^3+2M)^2}\Big)H_{z}\\
&&+\Big(\frac{6Mr(r-2M)}{(\omega^2r^3+2M)^2}+\frac{r(\mu r^2+M)}{\omega^2r^3+2M}-\frac{\mu r^2+M}{\omega^2r^2}\Big)\frac{dH_z}{dr},\nonumber\\
\tilde{H}_{rr}'&=&\Big(\frac{2r(M\mu r^2+\mu^2r^4-M(r-2M))}{(r-2M)^2(\omega^2r^3+2M)}+\frac{12M^2r}{(\omega^2r^3+2M)^2(r-2M)}-\frac{2\mu(\mu r+M)}{\omega^2(r-2M)^2}\Big)H_z\\
&&+\Big(\frac{2(\mu r^2+r-M)}{\omega^2r(r-2M)}-\frac{12Mr^2}{(\omega^2r^3+2M)^2}-\frac{2r^2(\mu r^2+r-M)}{(r-2M)(\omega^2r^3+2M)}\Big)H_{z}',\nonumber\\
\tilde{H}_{\theta\theta}'&=&-\frac{Mr^2}{\omega^2r^3+2M}H_z(r)-\frac{2M(r-2M)}{\omega^4r^3+2M\omega^2}\frac{dH_z}{dr}(r)
\end{eqnarray}
where the ODE~(\ref{HZ}) with~$\mathfrak{h}=H_z$ has been used. 
To determine the behavior of these new metric perturbation components close to the future event horizon~$\mathcal{H}^+_A$ one must substitute~$H^{2M,\pm}_z(r):=\mathfrak{h}^{2M,\pm}(r)$ from equations~(\ref{AHP})--(\ref{AHM}). Substituting~$H^{2M,\pm}_z(r):=\mathfrak{h}^{2M,\pm}(r)$ from equations~(\ref{basis1})  and~(\ref{basis2}) into these expressions gives leading order behavior close to the future event horizon~$\mathcal{H}_A^+$ determined by the relations
\begin{eqnarray}
\tilde{H}_{vv}^{2M,\pm}&=&f_{vv}(r)(r-2M)^{\pm 2M\mu},\\
\tilde{H}_{vr}^{2M,\pm}&=&\Big(\frac{(\mu\mp \mu)(1+4M\mu)}{2\omega^2(1+4M^2\omega^2)}(r-2M)^{-1}+f_{vr}(r)\Big)(r-2M)^{\pm 2M\mu},\\
\tilde{H}_{rr}^{2M,\pm}&=&\Big(\frac{-2(1\mp 1)M\mu(1+4M\mu)}{\omega^2(1+4M^2\omega^2)}(r-2M)^{-2}+k_{\pm}(r-2M)^{-1}+f_{rr}(r)\Big)(r-2M)^{\pm 2M\mu},\\
\tilde{H}_{\theta\theta}^{2M,\pm}&=&f_{\theta\theta}(r)(r-2M)^{\pm2M\mu}
\end{eqnarray}
with~$f_{vv}$,~$f_{vr}$,~$f_{rr}$,~$f_{\theta\theta}$ smooth functions of~$r\in [2M,\infty)$ which are non-vanishing at $2M$,~$k_+=0$ and~$k_-$ a non-zero constant depending on~$\omega,M$ and~$\mu$. Therefore, multiplying~$\tilde{H}_{vv}^{2M,+}$,~$\tilde{H}_{vr}^{2M,+}$,~$\tilde{H}_{rr}^{2M,+}$ and~$\tilde{H}_{\theta\theta}^{2M,+}$ by~${e^{\mu t}=e^{\mu v}e^{-\mu r}(r-2M)^{-2M\mu}}$ gives
\begin{eqnarray}
e^{\mu t+i\omega z}\tilde{H}_{vv}^{2M,+}&=&f_{vv}(r)e^{\mu v-\mu r+i\omega z},\\
e^{\mu t+i\omega z}\tilde{H}_{vr}^{2M,+}&=&f_{vr}(r)e^{\mu v-\mu r+i\omega z},\\
e^{\mu t+i\omega z}\tilde{H}_{rr}^{2M,+}&=&f_{rr}(r)e^{\mu v-\mu r+i\omega z},\\
e^{\mu t+i\omega z}\tilde{H}_{\theta\theta}^{2M,+}&=&f_{\theta\theta}(r)e^{\mu v-\mu r+i\omega z},
\end{eqnarray}
which can indeed be smoothly extended to the future event horizon~$\mathcal{H}^+_A$.
\end{proof}
\begin{remark}
The form of the pure gauge solution defined by equation~(\ref{GT}) can be derived as follows. From lemma~\ref{GML}, a mode solution $\tilde{h}$ of the form~(\ref{TTA}) satisfies the harmonic/transverse-traceless~(\ref{HG}) gauge conditions. Take a mode solution $h$ in spherical gauge~(\ref{gaugechoice}) add the pure gauge solution $h_{\mathrm{pg}}=2\nabla_{(a}\xi_{b)}$ for some vector field
\begin{eqnarray}
\xi=e^{\mu t+i\omega z}\zeta
\end{eqnarray}
where $\zeta$ is a vector field which depends only on $r$. From a direct calculation of $h+h_{\mathrm{pg}}$ one can see that, to obtain a solution $\tilde{h}$ of the form~(\ref{TTA}), $\zeta$ must be given by equations~(\ref{GT}). 
\end{remark}
\begin{remark}
To explicitly see the singular behaviour of the mode solution~$h^{\pm}$ in spherical gauge~(\ref{gaugechoice}) with~$\mu>0$ and~$\omega\neq 0$ associated, via proposition~\ref{construct}, to either $\mathfrak{h}^{2M,\pm}$, consider directly transforming to ingoing Eddington--Finkelstein coordinates. This transformation gives the following basis elements:
\begin{eqnarray}
{H_{rr}^{2M,\pm}}'&=&\Big(\frac{\partial t}{\partial r}\Big)^2H^{2M,\pm}_{t}+2\Big(\frac{\partial t}{\partial r}\Big)\mu H^{2M,\pm}_{v}+H^{2M,\pm}_{r}(r),\\
{H_{vv}^{2M,\pm}}'&=&{H_{t}^{2M,\pm}}(r),\\
{H_{vr}^{2M,\pm}}'&=&\Big(\frac{\partial t}{\partial r}\Big)H_{t}^{2M,\pm}(r)+\mu H_{v}^{2M,\pm}(r),\\
{H_{zz}^{2M,\pm}}'&=&H_z^{2M,\pm}(r),\label{HzzT}
\end{eqnarray}
where~$H_v^{2M,\pm}$,~$H_{t}^{2M,\pm}$ and~$H_{r}^{2M,\pm}$ are the basis for solutions for~$H_v$,~$H_t$ and~$H_r$ constructed from proposition~(\ref{construct}). These relevant expressions can be found from equations~(\ref{HR})--(\ref{HT}). \\

First, if~$4M\mu$ is a positive integer and the coefficent~$C_N$ does not vanish then, by equation~(\ref{HzzT}), the basis element~${H^{2M,-}_{zz}}'(r)=H_z^{2M,-}=\mathfrak{h}^{2M,-}$ has an essential logarithmic divergence and is therefore always singular at the future event horizon~$\mathcal{H}^+_A$.\\

If~$C_N=0$ or~$4M\mu$ is not a positive integer then the basis elements~$H_z^{2M,\pm}=\mathfrak{h}^{2M,\pm}$ are given by equations~(\ref{basis1}) and~(\ref{basis2}) with first order coefficients~(\ref{coeff1}). Substituting the basis into equations~(\ref{HR})--(\ref{HT}) for the other metric perturbation component, one finds
\begin{eqnarray}
H_{r}^{2M,\pm}&=&(r-2M)^{-2\pm 2M\mu}\Big(\frac{M^2( \pm 4M\mu-1)}{1+4M^2\omega^2}+\frac{M(4M^2(2\mu^2+\omega^2)\pm 6 M\mu-1)}{2(1+4M^2\omega^2)}(r-2M)\\&&+\mathcal{O}((r-2M)^2)\Big),\nonumber\\
H_{t}^{2M,\pm}&=&(r-2M)^{\pm 2M\mu}\Big(\frac{(1+ 4M\mu)( 4M\mu-1)}{4(1+4M^2\omega^2)}\\
&&+\frac{3+4M^2(8\mu^2-\omega^2)\pm2M\mu(8M^2(2\mu^2+\omega^2)-11)}{8M(1+4M^2\omega^2)}(r-2M)+\mathcal{O}((r-2M)^2)\Big),\nonumber\\
H_v^{2M,\pm}&=&(r-2M)^{-1+ 2M\mu}\Big(\frac{M^2(\pm 4M\mu-1)}{1+4M^2\omega^2}+\frac{M(2M^2(2\mu^2+\omega^2)-1\pm 5M\mu)}{1+4M^2\omega^2}(r-2M)\\
&&+\mathcal{O}((r-2M)^2)\Big)\nonumber.
\end{eqnarray}
Transforming to ingoing Eddington--Finkelstein coordinates gives
\begin{eqnarray}
{H_{rr}^{2M,\pm}}'&= &(r-2M)^{-2\pm 2M\mu}\Big(\frac{2M^2(1-2M\mu(1\mp 1))( 4M\mu-1)}{1+4M^2\omega^2}\\
&&+\frac{2M^2\mu \big((3\mp4)+2(9\mp 7)M\mu-(1\mp 1)4M^2(2\mu^2+\omega^2)\big)}{1+4M^2\omega^2}(r-2M) \nonumber\\
&&+\mathcal{O}((r-2M)^2)\Big),\nonumber\\
{H_{vv}^{2M,\pm}}'&= & (r-2M)^{\pm 2M\mu}\Big(\frac{(1+ 4M\mu)(4M\mu-1)}{4(1+4M^2\omega^2)}+\mathcal{O}(r-2M)\Big),\\
{H_{vr}^{2M,\pm}}'&=&(r-2M)^{-1\pm 2M\mu}\Big(\frac{M(2M\mu(1\mp 2)-1)(\pm 4M\mu -1)}{2(1+4M^2\omega^2)}+\mathcal{O}(r-2M)\Big),\\
{H_{zz}^{2M,\pm}}'&=&(r-2M)^{\pm 2M\mu}(1+\mathcal{O}(r-2M)).
\end{eqnarray}
Note that the full mode solution~$h$ constructed from proposition~\ref{construct} involves a factor of~$e^{\mu t}=e^{\mu v}e^{-\mu r}(r-2M)^{-2M\mu}$ so, after multiplication by this exponential factor, one can see that the basis elements~${H_{\mu\nu}^{2M,-}}'$ are always singular, i.e., a solution with~$k_2\neq 0$ is always singular at the future event horizon. The components~$e^{\mu t}{H_{vv}^{2M,+}}'$ and~$e^{\mu t}{H_z^{2M,+}}'$ are unconditionally smooth. However, in general, the components~$e^{\mu t}{H_{rr}^{2M,+}}'$ and~$e^{\mu t}{H_{vr}^{2M,+}}'$ remain singular at the future event horizon~$\mathcal{H}^+_A$ unless~$4M\mu=1$ or~$-2+2M\mu\in \mathbb{N}\cup \{0\}$ or~$-2+2M\mu>2$\footnote{In appendix \ref{bounds} it is shown that for existence of a solution~$\mathfrak{h}$ with~$\mu>0$ which has~$k_2=0$ and is finite at infinity (see section \ref{asympinf}) then~$\mu<\frac{3}{16M}\sqrt{\frac{3}{2}}<\frac{1}{4M}$.}. So \underline{neither} basis perturbation $h^{\pm}$ in spherical gauge~(\ref{gaugechoice}) extends, in general, smoothly across the future event horizon $\mathcal{H}^+_A$. 
\end{remark}

\subsubsection{Spacelike Infinity $i^0_A$}\label{asympinf}
The goal of this section is to identify the admissible boundary conditions for a solution~$\mathfrak{h}$ to the ODE~(\ref{HZ}) as~$r\rightarrow \infty$. This requires one to understand the behaviour as~$r\rightarrow \infty$ of the mode solution~$h$ in spherical gauge~(\ref{gaugechoice}) of the linearised vacuum Einstein equation~(\ref{LVEE}) which results (through the construction in proposition~\ref{construct}) from~$\mathfrak{h}$.\\

In this section, a basis for solution~$\mathfrak{h}^{\infty,\pm}$ associated to~$r\rightarrow \infty$ is constructed. This basis~$\mathfrak{h}^{\infty,\pm}$ captures the asymptotic behavior of any solution to the ODE~(\ref{HZ}) as~$r\rightarrow \infty$. In particular, as $r\rightarrow \infty$,~$\mathfrak{h}^{\infty,+}$ grows exponentially and~$\mathfrak{h}^{\infty,-}$ decays exponentially. It will be shown that after the addition of the pure gauge solution~$h_{\mathrm{pg}}$ defined in equations~(\ref{GT}) and~(\ref{pertt}),~$h+h_{\mathrm{pg}}$ is a mode solution in harmonic/transverse-traceless gauge~(\ref{HG}) to the linearised Einstein vacuum equation which is a linear combination of solutions which grow or decay exponentially as~$r\rightarrow \infty$. The admissible boundary condition will be that the solution should decay exponentially, from which it will follow that $\mathfrak{h}=a\mathfrak{h}^{\infty,-}$. \\

One should note that the functions $P_{\omega}(r)$ and $Q_{\omega}(r)-\frac{\mu^2r^2}{(r-2M)^2}$ admit convergent series expansions in a neighbourhood of~$r=\infty$
\begin{align}
P_{\omega}(r)=\sum_{n=0}^{\infty}\frac{p_n}{r^n}\qquad
Q_{\omega}(r)=\sum_{n=0}^{\infty}\frac{q_n}{r^n}
\end{align}
with~$p_0=0$,~$p_1=-4$,~$q_0=-(\omega^2+\mu^2)$ and~$q_1=-2M(\omega^2+2\mu^2)$. 
Therefore,~$r=\infty $ is an irregular singular point of the ODE~(\ref{HZ}) according to the discussion of appendix~\ref{Sing}. The equations~(\ref{lambdapm}) and~(\ref{mupm}) from appendix~\ref{Sing} give
\begin{eqnarray}
\lambda_{\pm}=\pm\sqrt{\mu^2+\omega^2},\quad \mu_{\pm}=2\pm \frac{M(2\mu^2+\omega^2)}{\sqrt{\mu^2+\omega^2}}.
\end{eqnarray}
From theorem~\ref{ThE3}, there exists a unique basis for solutions~$\mathfrak{h}^{\infty,\pm}(r)$ to the ODE~(\ref{HZ}) satisfying
\begin{eqnarray}
\mathfrak{h}^{\infty,\pm}=e^{\pm\sqrt{\mu^2+\omega^2}r}r^{2\pm \frac{M(2\mu^2+\omega^2)}{\sqrt{\mu^2+\omega^2}}}+\mathcal{O}\Big(e^{\pm\sqrt{\mu^2+\omega^2}r}r^{1\pm \frac{M(2\mu^2+\omega^2)}{\sqrt{\mu^2+\omega^2}}}\Big).\label{AI}
\end{eqnarray}
Therefore a general solution will be of the form
\begin{eqnarray}
\mathfrak{h}=c_1\mathfrak{h}^{\infty,+}+c_2\mathfrak{h}^{\infty,-}\label{basisinflincomb}
\end{eqnarray}
with~$ c_1,c_2\in \mathbb{R}$.
\begin{prop}\label{A2}
Let~$\mathfrak{h}$ be a solution to the ODE~(\ref{HZ}). Let~$h$ be the mode solution to the linearised vacuum Einstein equation~(\ref{LVEE}) in spherical gauge~(\ref{gaugechoice}) associated to the solution~$\mathfrak{h}$ and let~$h_{\mathrm{pg}}$ the pure gauge solution defined by equations~(\ref{GT}) and~(\ref{pertt}) such that~$h+h_{\mathrm{pg}}$ satisfies the harmonic/transverse-traceless gauge~(\ref{HG}) conditions. Then the solution~$h+h_{\mathrm{pg}}$ to the ODE~(\ref{HZ}) decays exponentially towards spacelike infinity~$i^0_A$ if~$c_1=0$, where $c_1$ is defined by equation~(\ref{basisinflincomb}).
\end{prop}
\begin{proof}
Defining~$H^{\infty,\pm}_z(r):=\mathfrak{h}^{\infty,\pm}(r)$ and using equations~(\ref{trans1})--(\ref{trans4}) one can construct the corresponding basis for solutions as~$\tilde{H}_{tt}$,~$\tilde{H}_{tr}$,~$\tilde{H}_{rr}$ and~$\tilde{H}_{\theta\theta}$ from proposition~\ref{construct}. Note that equations~(\ref{trans1})--(\ref{trans4}) define the components of the mode solution~$h+h_{\mathrm{pg}}$ to the linearised vacuum Einstein equation~(\ref{LVEE}) which satisfies harmonic/transverse-traceless gauge~(\ref{HG}). Asymptotically~$\tilde{H}_{tt}$,~$\tilde{H}_{tr}$,~$\tilde{H}_{rr}$ and~$\tilde{H}_{\theta\theta}$ have the following behavior:
\begin{eqnarray}
H_{tt}^{\infty,\pm}&=&e^{\pm\sqrt{\mu^2+\omega^2}r}r^{-1\pm \frac{M(2\mu^2+\omega^2)}{\sqrt{\mu^2+\omega^2}}}+\mathcal{O}\Big(e^{\pm\sqrt{\mu^2+\omega^2}r}r^{-2\pm \frac{M(2\mu^2+\omega^2)}{\sqrt{\mu^2+\omega^2}}}\Big),\\
H_{tr}^{\infty,\pm}&=&e^{\pm\sqrt{\mu^2+\omega^2}r}r^{-1\pm \frac{M(2\mu^2+\omega^2)}{\sqrt{\mu^2+\omega^2}}}+\mathcal{O}\Big(e^{\pm\sqrt{\mu^2+\omega^2}r}r^{-2\pm \frac{M(2\mu^2+\omega^2)}{\sqrt{\mu^2+\omega^2}}}\Big),\\
H_{rr}^{\infty,\pm}&=&e^{\pm\sqrt{\mu^2+\omega^2}r}r^{-1\pm \frac{M(2\mu^2+\omega^2)}{\sqrt{\mu^2+\omega^2}}}+\mathcal{O}\Big(e^{\pm\sqrt{\mu^2+\omega^2}r}r^{-2\pm \frac{M(2\mu^2+\omega^2)}{\sqrt{\mu^2+\omega^2}}}\Big),\\
H_{\theta\theta}^{\infty,\pm}&=&e^{\pm\sqrt{\mu^2+\omega^2}r}r^{1\pm \frac{M(2\mu^2+\omega^2)}{\sqrt{\mu^2+\omega^2}}}+\mathcal{O}\Big(e^{\pm\sqrt{\mu^2+\omega^2}r}r^{\pm \frac{M(2\mu^2+\omega^2)}{\sqrt{\mu^2+\omega^2}}}\Big).
\end{eqnarray}
It is clear from these expressions that, if~$c_1=0$, then the mode solution~$h+h_{\mathrm{pg}}$ decays exponentially as~$r\rightarrow\infty$.
\end{proof}

\subsection{Reduction of the Proof of Theorem \ref{RT}}\label{Reduce}
This section summarises propositions \ref{construct}, \ref{notpuregaugeS}, \ref{A1} and \ref{A2} to give a full description of the permissible asymptotic behavior of a mode solution~$h$ in spherical gauge~(\ref{gaugechoice}) which is not pure gauge. This provides a reduction of theorem \ref{RT} to proving that there exists a solution~$\mathfrak{h}$ to the ODE~(\ref{HZ}) which has~$\mu>0$,~$\omega\neq 0$ and obeys the admissible boundary conditions:~$k_2=0$ and~$c_1=0$.

\begin{prop}\label{Asymp}
Let~$\mu> 0$ and~$\omega\in \mathbb{R}$ with~$\omega\neq 0$. Let~$\mathfrak{h}^{2M,\pm}$ be the basis for the space of solutions to the ODE~(\ref{HZ}) as defined in equations~(\ref{AHP}) and~(\ref{AHM}) and~$\mathfrak{h}^{\infty,\pm}$ be the basis for the space of solutions to the ODE~(\ref{HZ}) as defined in equation~(\ref{AI}). In particular, to any solution~$\mathfrak{h}$ of the ODE~(\ref{HZ}) one can ascribe four numbers~$k_1,k_2,c_1,c_2\in \mathbb{R}$ defined by
\begin{eqnarray}
\mathfrak{h}(r)&=&k_1\mathfrak{h}^{2M,+}(r)+k_2\mathfrak{h}^{2M,-}(r),\\
\mathfrak{h}(r)&=&c_1\mathfrak{h}^{\infty,+}(r)+c_2\mathfrak{h}^{\infty,-}(r).
\end{eqnarray}
Let~$h$ be the mode solution in spherical gauge~(\ref{gaugechoice}) to the linearised vacuum Einstein equation~(\ref{LVEE}) on the exterior~$\mathcal{E}_A$ of Schwarzschild black string~$\mathrm{Sch}_4\times \mathbb{R}$ associated to~$\mathfrak{h}$ via proposition~\ref{construct}. Let~$h_{\mathrm{pg}}$ be the pure gauge solution as defined in equations~(\ref{GT}) and~(\ref{pertt}). Then~$h+h_{\mathrm{pg}}$ decays exponentially towards spacelike infinity~$i^0_A$ and is smooth at the future event horizon~$\mathcal{H}^+_A$ if~$k_2=0$ and~$c_1=0$. Moreover, $h+h_{\mathrm{pg}}$ satisfies the harmonic/transverse-traceless gauge conditions~(\ref{HG}) and cannot be a pure gauge solution.\\

Under the additional assumption that~$\omega R\in \mathbb{Z}$ the mode solution~$h$ defined above can be interpreted as a mode solution to the linearised vacuum Einstein equation~(\ref{LVEE}) on the exterior~$\mathcal{E}_A$ of Schwarzschild black string~$\mathrm{Sch}_4\times \mathbb{S}_R^1$. Hence, if~$\omega R\in \mathbb{Z}$ the above statement applies to the exterior~$\mathcal{E}_A$ of~$\mathrm{Sch}_4\times \mathbb{S}^1_{R}$.
\end{prop}

The next section (see proposition \ref{propG}) will prove the existence of a solution~$\mathfrak{h}$ to the ODE~(\ref{HZ}) satisfying the properties of proposition~\ref{Asymp}. In particular, for all~${|\omega|\in[\frac{3}{20M},\frac{8}{20M}]}$, a solution~$\mathfrak{h}$ to the ODE~(\ref{HZ}) with~$\mu>\frac{1}{40\sqrt{10}M}>0$,~$k_2=0$ and~$c_1=0$ is constructed. If~$R>4M$, then there exists an integer~${n\in [\frac{3R}{20M},\frac{8R}{20M}]}$. Hence, one can choose~$\omega~$ such that the constructed~$\mathfrak{h}$ gives rise to a mode solution on~$\mathrm{Sch}_4\times \mathbb{S}^1_R$. Moreover, on~$\mathrm{Sch}_4\times \mathbb{S}^1_R$,~$h$ will manifestly have finite energy in the sense that~$||h|_{\Sigma}||_{H^1}$ and~$||\partial_{t_*} h|_{\Sigma}||_{L^2}$ are finite\footnote{On~$\mathrm{Sch}_4\times \mathbb{R}$,~$h$ will not have finite energy due to the periodic behaviour in $z$ on $\mathbb{R}$}. Thus, theorem~1.1 follows from proposition~\ref{Asymp} and proposition~\ref{propG}. 

\pagebreak
\section{The Variational Argument}\label{BSU}
 By proposition~\ref{Asymp}, the proof of theorem \ref{RT} has now been reduced to exhibiting a solution~$ \mathfrak{h}$ to~(\ref{HZ}) with~$\mu>0$,~$\omega\neq 0$,~$k_2=0$ and~$c_1=0$. This section establishes the required proposition thus completing the proof.
\begin{prop}\label{propG}
For all~$|\omega|\in[\frac{3}{20M},\frac{8}{20M}]$ there exists a~$C^{\infty}((2M,\infty))$ solution~$\mathfrak{h}$ to ODE~(\ref{HZ}) with~$\mu>0$, and in the language of proposition~\ref{Asymp},~$k_2=0$ and~$c_1=0$.
 \end{prop}
  In order to exhibit such a solution~$\mathfrak{h}$ to the ODE~(\ref{HZ}), it is convenient to rescale the solution and change coordinates in the ODE~(\ref{HZ}) so as to recast as a Schr\"odinger equation for a function $u$. This transformation is given in section~\ref{AODE}. In section~\ref{DVA} an energy functional is assigned to the resulting Schr\"odinger operator. With the use of a test function (constructed in section~\ref{TF}), a direct variational argument can be run to establish that for~$|\omega|\in[\frac{3}{20M},\frac{8}{20M}]$, there exists a weak solution~$u\in H^1(\mathbb{R})$ with~$||u||_{H^1(\mathbb{R})}=1$ such that~$\mu>0$. The proof of proposition~\ref{propG} concludes by showing that the solution~$u$ is indeed smooth for~$r\in (2M,\infty)$ and satisfies the conditions of proposition~\ref{Asymp}, i.e.,~$k_2=0$ and~$c_1=0$. 
\subsection{Schr\"odinger Reformulation}\label{AODE}
To reduce the number of parameters in the ODE~(\ref{HZ}), one can eliminate the mass parameter with ${x:=\frac{r}{2M}}$,~${\hat{\mu}:=2M\mu}$ and~${\hat{\omega}:=2M\omega}$ to find
\begin{eqnarray}
\frac{d^2\mathfrak{h}}{dx^2}(x)+p_{\hat{\omega}}(x)\frac{d\mathfrak{h}}{dx}+\Big(q_{\hat{\omega}}(x)-\frac{\hat{\mu}^2x^2}{(x-1)^2}\Big)\mathfrak{h}(x)=0,\label{HZX}
\end{eqnarray}
with
\begin{eqnarray}
p_{\hat{\omega}}(x)&=&\frac{1}{x-1}-\frac{5}{x}+\frac{6}{x(\hat{\omega}^2x^3+1)},\\
q_{\hat{\omega}}(x)&=&\frac{3}{x^2(x-1)}-\frac{\hat{\omega}^2x}{x-1}-\frac{3}{x^2(x-1)(1+\hat{\omega}^2x^3)}.
\end{eqnarray}

 Following proposition~\ref{Sch} from appendix~\ref{SchT} one can now transform the equation~(\ref{HZX}) into regularised Schr\"odinger form by introducing a weight function~$\mathfrak{h}(x)=w(x)\tilde{\mathfrak{h}}(x)$ and changing coordinates to~$x_*=\frac{r_*}{2M}=x+\log|x-1|$. This will produce a Schr\"odinger operator with a potential which decays to zero at the future event horizon and tends to the constant $\hat{\omega}^2$ at spatial infinity. From proposition~\ref{Sch} the weight function must satisfy the ODE
\begin{eqnarray}
\frac{dw}{dx}+\frac{(1-2\omega^2x^3)}{x(1+\omega^2x^3)}w=0.
\end{eqnarray}
 The desired solution for the weight function is
\begin{eqnarray}
w(x)=\frac{(1+\hat{\omega}^2x^3)}{x}.
\end{eqnarray}
The ODE~(\ref{HZX}) becomes
\begin{eqnarray}
-\frac{d^2\tilde{\mathfrak{h}}}{dx_*^2}(x_*)+V(x_*)\tilde{\mathfrak{h}}(x_*)=-\hat{\mu}^2\tilde{\mathfrak{h}}(x_*),\label{ODEF}
\end{eqnarray}
where~$V:\mathbb{R}\rightarrow \mathbb{R}$ can be found from equation~(\ref{PotT}) to be 
\begin{eqnarray}
 V(x_*)=\hat{\omega}^2\frac{(x-1)}{x}+\frac{(6x-11)(x-1)}{x^4}+\frac{18(x-1)^2}{x^4(1+\hat{\omega}^2x^3)^2}-\frac{6(4x-5)(x-1)}{x^4(1+\hat{\omega}^2x^3)},\quad x\in (1,\infty),\label{pot}
 \end{eqnarray}
 where~$x$ is understood as an implicit function of~$x_*$.\\
 
 As a trivial consequence of proposition~\ref{Asymp} in section~\ref{RC} on asymptotics of the solution to the ODE~(\ref{HZ}), one has the following proposition for the asymptotics of the Schr\"odinger equation~(\ref{ODEF}). 
\begin{prop}\label{Asymp2}
Assume~$\hat{\mu}>0$. To any solution~$\tilde{\mathfrak{h}}$ to the Schr\"odinger equation~(\ref{ODEF}) one can ascribe four numbers~$\tilde{k}_1,\tilde{k}_2,\tilde{c}_1,\tilde{c}_2\in \mathbb{R}$ defined by
\begin{eqnarray}
\tilde{\mathfrak{h}}(x_*)&=&\tilde{k}_1\tilde{\mathfrak{h}}^{2M,+}(x_*)+\tilde{k}_2\tilde{\mathfrak{h}}^{2M,-}(x_*)\quad \text{as}\quad x_*\rightarrow -\infty,\\
\tilde{\mathfrak{h}}(x_*)&=&\tilde{c}_1\tilde{\mathfrak{h}}^{\infty,+}(x_*)+\tilde{c}_2\tilde{\mathfrak{h}}^{\infty,-}(x_*)\quad \text{as}\quad x_*\rightarrow \infty
\end{eqnarray}
with
\begin{eqnarray}
\tilde{\mathfrak{h}}^{2M,\pm}&:=&\frac{\mathfrak{h}^{2M,\pm}}{w}\label{AHT},\\
\tilde{\mathfrak{h}}^{\infty,\pm}&:=&\frac{\mathfrak{h}^{\infty,\pm}}{w}\label{AIT}.
\end{eqnarray}
The conditions that~$\tilde{c}_1=0$ and~$\tilde{k}_2=0$ are equivalent to, in the language of proposition~\ref{Asymp},~$c_1=0$ and~$k_2=0$. 
\end{prop}
\begin{remark} In the case $4M\mu$ is not a positive integer or $4M\mu$ is a positive integer and $C_N=0$ the leading order terms of these basis elements are
\begin{eqnarray}
\tilde{\mathfrak{h}}^{2M,\pm}&=&(x-1)^{\pm \hat{\mu}}\Big(\frac{1}{1+\hat{\omega}^2}+\mathcal{O}(x-1)\Big),\label{asymschh}\\
\tilde{\mathfrak{h}}^{\infty,\pm}&=&e^{\pm\sqrt{\hat{\mu}^2+\hat{\omega}^2}x}x^{ \pm\frac{(2\hat{\mu}^2+\hat{\omega}^2)}{2\sqrt{\hat{\mu}^2+\hat{\omega}^2}}}\Big(\frac{1}{\hat{\omega}^2}+\mathcal{O}\Big(\frac{1}{x}\Big)\Big).\label{asymschi}
\end{eqnarray}
\end{remark}
\subsection{Direct Variational Argument}\label{DVA}
This section establishes a variational argument which will be used to infer the existence of a negative eigenvalue to the Schr\"odinger operator in equation~(\ref{ODEF}).
 \begin{prop}\label{propvar}
 Let~$W:\mathbb{R}\rightarrow \mathbb{R}$ and define
 \begin{eqnarray}
 E_0:=\inf_{\underset{||v||_{L^2(\mathbb{R})}=1}{v\in H^1(\mathbb{R})}}\Big\{E(v):=\langle\nabla v,\nabla v\rangle_{L^2(\mathbb{R})}+\langle Wv,v\rangle_{L^2(\mathbb{R})}\Big\}.
 \end{eqnarray}
 Suppose that $W=p+q$ with~$q\in C^0(\mathbb{R})$ such that
 \begin{eqnarray}
 \lim_{|x|\rightarrow \infty}q(x)=0
 \end{eqnarray}
 and~$p(x)\in L^{\infty}(\mathbb{R})$ positive. If~$E_0<0$, then there exists~$u\in H^1(\mathbb{R})$ such that~$||u||_{L^2(\mathbb{R})}=1$ and~${E(u)=E_0}$.
 \end{prop}
 \begin{proof}
 By the definition of the infimum there exists a minimizing sequence~$(u_m)_m\subset H^1(\mathbb{R})$ and~$||u_m||_{L^2}=1$ such that 
 \begin{eqnarray}
 \lim_{n\rightarrow \infty}E(u_n)=E_0.
 \end{eqnarray}
Now,~$u_n$ are bounded in~$H^1(\mathbb{R})$ by the following argument. There exists an~$M\in \mathbb{N}$ such that, for all~$m\geq M$,
 \begin{eqnarray}
 E(u_m)\leq E_0+1.
 \end{eqnarray}
 So, for~$m\geq M$,
 \begin{eqnarray}
\langle\nabla u_m,\nabla u_m\rangle_{L^2(\mathbb{R})}\leq E_0+1+\sup_{x\in \mathbb{R}}|p(x)|+\sup_{x\in \mathbb{R}}|q(x)|.
 \end{eqnarray}
 Hence~$||u_m||_{H^1(\mathbb{R})}$ is controlled. Now using theorem~\ref{lc} from appendix~\ref{analysis} there exists a subsequence~$(u_{m_n})_n$ such that~$u_{m_n}\rightharpoonup u$ in~$H^1(\mathbb{R})$. \\
 
 \noindent Consider
\begin{eqnarray}
E(u_m)=\int_{\mathbb{R}}|\nabla u_m|^2+p(x)|u_m|^2+q(x)|u_m|^2dx.\label{Es}
\end{eqnarray}
Since the Dirichlet energy is lower semicontinuous, only the latter two terms under the integral~(\ref{Es}) need to be examined more carefully. The middle term in integral~$(\ref{Es})$ are simply a weighted~$L^2$ integral, so lower semicontinuity is established via
\begin{eqnarray}
||u_n-u||^2_{L^2_{p}}=\langle u_n-u,u_n-u\rangle_{L^2_{p}}=||u_n||^2_{L^2_{q}}-2\langle u_n,u\rangle_{L^2_{p}}+||u||^2_{L^2_{p}}.
\end{eqnarray}
So
\begin{eqnarray}
||u||^2_{L^2_p}\leq ||u_n||^2_{L^2_p}-2\langle u,u_n-u\rangle_{L^2_p}.
\end{eqnarray}
Hence by weak convergence
\begin{eqnarray}
||u||^2_{L^2_p}\leq \liminf_{n\rightarrow \infty}||u_n||^2_{L^2_p}.
\end{eqnarray}
The proposition~\ref{compactness} from appendix~\ref{analysis} establishes that the multiplication operator~$M_q:u\rightarrow qu$ is compact from~$H^1(\mathbb{R})$ to~$L^2(\mathbb{R})$. Hence, by the characterisation of compactness through weak convergence (theorem~\ref{lc} from appendix~\ref{analysis}),~$qu_m\rightarrow qu$ in~$L^2(\mathbb{R})$. Therefore
\begin{eqnarray}
\langle qu,u\rangle_{L^2}=\lim_{m\rightarrow \infty}\langle qu_m,u_m \rangle_{L^2}=\liminf_{m\rightarrow\infty}\langle qu_m,u_m \rangle_{L^2}.
\end{eqnarray}
Hence, the last term under the integral~(\ref{Es}) is also lower semicontinuous. Therefore
\begin{eqnarray}
E(u)\leq \liminf_{n\rightarrow \infty}E(u_n)=E_0.
\end{eqnarray}
\noindent Since the infimum is negative the minimiser is non-trivial. One needs to show that there is no loss of mass, i.e.,~$||u||_{L^2}=1$. Note~$||u||_{L^2}\leq \liminf_{n\rightarrow \infty}||u_n||_{L^2}=1$. So suppose~$||u||_{L^2}<1$ and define~$\tilde{u}=\frac{u}{||u||_{L^2}}$ so~$||\tilde{u}||_{L^2}=1$, then
 \begin{eqnarray}
 E(\tilde{u})=\frac{E_0}{||u||^2_{L^2(\mathbb{R})}}\leq E_0
 \end{eqnarray}
 since~$||u ||_{L^2}\leq 1$. Hence one would obtain a contradition to the infimum if~$||u||_{L^2}<1$.
 \end{proof}
 \begin{corollary}\label{corm}
 Let $W=V$ with $V$ as defined in equation~(\ref{pot}) then
 \begin{eqnarray}
 E(v):=\langle \nabla v,\nabla v\rangle_{L^2(\mathbb{R})}+\langle V v,v \rangle_{L^2(\mathbb{R})}\label{EFM}\qquad E_0:=\inf_{\underset{||v||_{L^2(\mathbb{R})}=1}{v\in H^1(\mathbb{R})}}E(v)
 \end{eqnarray}
 satisfies the assumptions of proposition~\ref{propdirect}.
 \end{corollary}
 \begin{proof}
The function~$V:\mathbb{R}\rightarrow \mathbb{R}$ can be written as $V=p+q$ with $p$ and $q$ as follows. Define
 \begin{eqnarray}
p(x_*)&:= & \hat{\omega}^2\frac{x-1}{x},\\
q(x_*)&:=&\frac{(6x-11)(x-1)}{x^4}+\frac{18(x-1)^2}{x^4(1+\hat{\omega}^2x^3)^2}-\frac{6(4x-5)(x-1)}{x^4(1+\hat{\omega}^2x^3)}
 \end{eqnarray}
 where in these expressions $x$ considered as a implicit function of $x_*$. Since~$x\in (1,\infty)$, it follows that~$p(x_*)>0$ for all~$x_*\in \mathbb{R}$. Moreover,
 \begin{eqnarray}
 \sup_{x_*\in \mathbb{R}}|p(x_*)|=1.
 \end{eqnarray}
Therefore,~$p\in L^{\infty}(\mathbb{R})$. Note that the function~$q$ satisfies~$\lim_{|x_*|\rightarrow \infty}q(x_*)=0$. So the assumptions of proposition~\ref{propvar} hold.
 \end{proof}
 \subsection{The Test Function and Existence of a Minimiser}\label{TF}
The ODE~(\ref{ODEF}) is now in a form where a direct variational argument can be used to prove that there exists an eigenfunction of the Schr\"odinger operator associated to the left-hand side of the ODE~(\ref{ODEF}) with a negative eigenvalue, i.e.~$-\hat{\mu}^2<0$. The following proposition constructs a suitable test function such that it is in the correct function space,~$H^1(\mathbb{R})$, and, for all $|\hat{\omega}|\in[\frac{3}{10},\frac{8}{10}]$, implies that the infimum of the energy functional in equation~(\ref{EFM}) is negative. (As will be apparent, the negativity is inferred via complicated but purely algebraic calculations.)
\begin{prop}\label{h2proof}
Define~$u_T(x_*):=x(1+|\hat{\omega}|^2x^3)(x-1)^{\frac{1}{n}}e^{-4|\hat{\omega}| (x-1)}$ where~$x$ is an implicit function of~$x_*$,~$n$ is a finite non-zero natural number,~$\hat{\omega}\in \mathbb{R}\setminus\{0\}$ and define $E$ and $E_0$ as in equation~(\ref{EFM}) of corollary~\ref{corm}. Then~$u_T\in H^1(\mathbb{R})$ and, for~$n=100$ and~$|\hat{\omega}|\in[\frac{3}{10},\frac{8}{10}]$,~${E_0\leq E(u_T)<0}$. 
\end{prop}
\begin{proof}
Let $k\in \mathbb{N}\cup\{0\}$ and define the following functions
\begin{align}
f_k(x):=x^{k-1}(x-1)^{\frac{2}{n}-1}e^{-8|\hat{\omega}| (x-1)}\label{fk}.
\end{align}
The $H^1(\mathbb{R})$-norm of $u_T$ can be expressed as
\begin{eqnarray}
||u_{T}||_{H^1(\mathbb{R})}^2=\int_{1}^{\infty}\Big|\frac{x-1}{x} \frac{du_T}{dx}\Big|^2\frac{x}{x-1}dx+\int_{1}^{\infty}|u_T|^2\frac{x}{x-1}dx
\end{eqnarray}
where on the right-hand side the change of variables from $x_*\in \mathbb{R}$ to $x\in(1,\infty)$ has been made. 
To calculate $||u_T||_{L^2(\mathbb{R})}$ it is useful to write it as a linear combination of the functions $f_k$ in equation~(\ref{fk}). Explicitly, one can show that 
\begin{eqnarray}
|u_T|^2\frac{x}{x-1}=f_4(x)+2|\hat{\omega}|^2f_7(x)+|\hat{\omega}|^4f_{10}(x).
\end{eqnarray}
Similarly, one can show that
\begin{align}
\Big|\frac{x-1}{x}\frac{du_T}{dx}\Big|^2\frac{x}{x-1}=\sum_{k=1}^{11}c_{k}f_{k-1}(x)
\end{align}
with
\begin{align}
\begin{split}
c_1&=1,\qquad c_2=-2 - \frac{2}{n} - 8 |\hat{\omega}|,\qquad c_3= 1 + \frac{1}{n^2} + \frac{2}{n} + 16 |\hat{\omega}| + \frac{8 |\hat{\omega}|}{n} + 
 16 |\hat{\omega}|^2,\\
c_4&= -8 |\hat{\omega}| - \frac{8 |\hat{\omega}|}{n} - 24 |\hat{\omega}|^2,\qquad c_5=-\frac{10 |\hat{\omega}|^2}{n} - 
 40 |\hat{\omega}|^3, \\
 c_6&=8 |\hat{\omega}|^2 + \frac{2 |\hat{\omega}|^2}{n^2} + \frac{10 |\hat{\omega}|^2}{n} + 80 |\hat{\omega}|^3 + \frac{
 16 |\hat{\omega}|^3}{n} + 32 |\hat{\omega}|^4, \qquad c_7=-40 |\hat{\omega}|^3 - \frac{16 |\hat{\omega}|^3}{n} - 
 48 |\hat{\omega}|^4, \\
 c_8&=-\frac{8 |\hat{\omega}|^4}{n} - 32 |\hat{\omega}|^5, \qquad c_9=16 |\hat{\omega}|^4 + \frac{|\hat{\omega}|^4}{n^2} + \frac{8 |\hat{\omega}|^4}{n} + 64 |\hat{\omega}|^5 +\frac{ 
 8 |\hat{\omega}|^5}{n} + 16 |\hat{\omega}|^6,\\
 c_{10}&=-32 |\hat{\omega}|^5 - \frac{8 |\hat{\omega}|^5}{n} - 
 32 |\hat{\omega}|^6, \qquad c_{11}=16 |\hat{\omega}|^6.
 \end{split}
\end{align}
One can express $E(u_T)$ with the change of variables from $x_*$ to $x$ as
\begin{eqnarray}
E(u_T)=\int_1^{\infty}\Big(\Big|\frac{x-1}{x}\frac{du_T}{dx}\Big|^2+V(u_T)^2\Big)\frac{x}{x-1}dx.
\end{eqnarray}
The integrand can be written as
\begin{eqnarray}
\Big(\Big|\frac{x-1}{x}\frac{du_T}{dx}\Big|^2+V(u_T)^2\Big)\frac{x}{x-1}=\sum_{k=1}^{11}a_kf_{k-1}(x),
\end{eqnarray}
with
\begin{align}
\begin{split}
a_1&=0,\qquad a_2=-\frac{2+n+8n|\hat{\omega}|}{n},\qquad
a_3=1 + \frac{1}{n^2} + 16 |\hat{\omega}| + 16 |\hat{\omega}|^2 + \frac{2 + 8 |\hat{\omega}|}{n},\\
a_4&=-|\hat{\omega}|\Big(\frac{8 (1 + n) }{n}+33 |\hat{\omega}|\Big),\qquad
a_5=\frac{(21 n-10 ) |\hat{\omega}|^2}{n} - 40 |\hat{\omega}|^3,\\
a_6&=|\hat{\omega}|^2\Big(\frac{2 (1 + 5 n - 2 n^2) }{n^2} + \frac{16 (1 + 5 n) |\hat{\omega}|}{n} + 
 32 |\hat{\omega}|^2\Big),\qquad
a_7=-|\hat{\omega}|^3\Big(\frac{8 (2 + 5 n) }{n} + 39 |\hat{\omega}|\Big),\\
a_8&=-|\hat{\omega}|^4\Big(\frac{(8 + 15 n) }{n} + 32 |\hat{\omega}|\Big),\quad
a_9=|\hat{\omega}|^4\Big(\frac{1 + 8 n + 22 n^2 }{n^2} + \frac{8 (1 + 8 n) }{n} + 
 16 |\hat{\omega}|^2\Big),\\
a_{10}&=-|\hat{\omega}|^5\Big(\frac{8 (1 + 4 n) }{n} + 33 |\hat{\omega}|\Big),\quad
a_{11}=17 |\hat{\omega}|^6.
\end{split}
\end{align}
Therefore, if one can compute the integrals
\begin{eqnarray}
I_{k}:=\int_{1}^{\infty }f_{k}(x)dx\label{IK0}
\end{eqnarray}
for $k=0,...,10$ then one can compute $||u_T||_{L^2(\mathbb{R})}$, $||\frac{du_T}{dx_*}||_{L^2(\mathbb{R})}$ and $E(u_T)$.\\

Defining a change variables in the integrals~(\ref{IK0}) by $t=x-1$, the integrals~(\ref{IK0}) become 
\begin{eqnarray}
I_k=\int_0^{\infty}(t+1)^{k-1}t^{\frac{2}{n}-1}e^{-8|\hat{\omega}| t}.\label{IK}
\end{eqnarray}
Note that the confluent hypergeometric function of the second kind $U(a,b;z)$ can be defined as
\begin{eqnarray}
U(a,b;z):=\frac{1}{\Gamma(a)}\int_0^{\infty}(t+1)^{b-a-1}t^{a-1}e^{-z t}
\end{eqnarray}
for $a,b,z\in \mathbb{C}$ with $\mathrm{Re}(a)>0$ and $\mathrm{Re}(z)>0$ where $\Gamma(a)$ is the Euler Gamma function, which can be defined through the integral
\begin{eqnarray}
\Gamma(a)=\int_{0}^{\infty}t^{a-1}e^{-t}dt
\end{eqnarray}
for $a\in \mathbb{C}$ with $\mathrm{Re}(a)>0$. For a reference see chapter 9 of~\cite{dover}. Therefore, setting $a=\frac{2}{n}$, $b=k+\frac{2}{n}$ and $z=8|\hat{\omega}|$ gives
\begin{align}
I_k=\Gamma\Big(\frac{2}{n}\Big)U\Big(\frac{2}{n},k+\frac{2}{n};8|\hat{\omega}|\Big).
\end{align}
The function $U(a,b;z)$ satisfies the following recurrence properties (see  chapter 9 of~\cite{dover} and chapter 16 of~\cite{HSP}):
\begin{align}
U(0,b;z)&=1\label{HG0}\\
U(a,b;z)-z^{1-b}U(1+a-b,2-b;z)&=0\label{HG1}\\
U(a,b;z)-aU(a+1,b;z)-U(a,b-1;z)&=0\label{HG2}\\
(b-a-1)U(a,b-1;z)+(1-b-z)U(a,b;z)+z U(a,b+1;z)&=0.\label{HG3}
\end{align}
Setting $a=\frac{2}{n}$, $b=1+\frac{2}{n}$ and $z=8|\hat{\omega}|$ in equation~(\ref{HG1}), and using equation~(\ref{HG0}) allows one to calculate $I_1$. Setting $a=\frac{2}{n}$, $b=2+\frac{2}{n}$ and $z=8|\hat{\omega}|$ in equation~(\ref{HG2}), using $I_1$ and equation~(\ref{HG0}) allows one to calculate $I_2$. Setting $a=\frac{2}{n}$, $b=k+\frac{2}{n}$ and $z=8|\hat{\omega}|$ in equation~(\ref{HG3}), using $I_{k-1},...,I_1$ and equation~(\ref{HG0}) allows one to calculate $I_k$. Finally, one can show that $I_0<\infty$ by the following argument. One can see from the definition of $I_k$ in equation~(\ref{IK}) that
\begin{eqnarray}
I_0=\int_1^{\infty}\frac{1}{x(x-1)}(x-1)^{\frac{2}{n}}e^{-8|\hat{\omega}| (x-1)}dx.
\end{eqnarray}
Now, since $e^{-8|\hat{\omega}|(x-1)}<1$ on $x\in (1,\infty)$ and $\frac{(x-1)^{\frac{2}{n}-1}}{x}<\frac{1}{2}(x-1)$ for $n\geq 1$ on $x\in(2,\infty)$, 
\begin{align}
I_0&\leq \int_1^2\frac{1}{x(x-1)}(x-1)^{\frac{2}{n}}+\frac{1}{2}\int_2^{\infty}(x-1)e^{-8|\hat{\omega}| (x-1)}<\infty.\label{I0}
\end{align}

Using the recurrence properties in equations~(\ref{HG0})--(\ref{HG3}) and the estimate~(\ref{I0}) allows one to explicitly show that ${||u_T||_{H^1(\mathbb{R})}<\infty}$ for~$n\geq 1$,~$\hat{\omega}\in \mathbb{R}\setminus\{0\}$, i.e.,~$u_T\in H^1(\mathbb{R})$. Moreover, one can calculate~$\frac{E(u_T)}{||u_T||_{L^2(\mathbb{R})}}$. Explicitly, $\frac{E(u_T)}{||u_T||_{L^2(\mathbb{R})}}$ is given by
\begin{eqnarray}
\frac{E(u_T)}{||u_T||^2_{L^2(\mathbb{R})}}=\frac{|\hat{\omega}|^2\sum_{i=1}^{9}p_i(n)|\hat{\omega}|^{i-1}}{\sum_{j=1}^{10}q_j(n)|\hat{\omega}|^{i-1}}
\end{eqnarray}
with
\begin{align}
\begin{split}
p_1(n)&:=16 + 416 n + 5576 n^2 + 36176 n^3 + 123809 n^4 + 234794 n^5 + 
 244459 n^6 + 128034 n^7 + 25560 n^8\\
 p_2(n)&:=32 n (16 + 336 n + 3296 n^2 + 15572 n^3 +29107 n^4 +21238 n^5+4361 n^6 - 366 n^7)\\
 p_3(n)&:=128 n^2 (56 + 924 n + 6130 n^2 + 20133 n^3 + 11972 n^4 - 3365 n^5 - 
   466 n^6)\\
 p_4(n)&:=1024 n^3 (56 + 700 n + 2750 n^2 + 6041 n^3- 1715 n^4 - 18 n^5)\\
 p_5(n)&:=2048 n^4 (140 +1260 n + 2225 n^2 + 3443 n^3 - 1758 n^4)\\
 p_6(n)&:=32768 n^5 (28 + 168 n + 43 n^2 + 111 n^3)\\
  p_7(n)&:=917504 n^6 (2 + 7 n - 3 n^2)\\
   p_8(n)&:=1048576 n^7 (2 + 3 n)\\
   p_9(n)&:=1048576 n^8\\
  q_1(n)&:=116 + 288 n + 2184 n^2 + 9072 n^3 + 22449 n^4 + 33642 n^5 + 
 29531 n^6 + 13698 n^7 + 2520 n^8\\
 q_2(n)&:=4 n (144 + 2016 n + 12104 n^2 + 39120 n^3 + 71801 n^4 + 73494 n^5 + 
   38171 n^6 + 7590 n^7)\\
 q_3(n)&:=128 n^2 (72 + 756 n + 3534 n^2 + 8535 n^3 + 11180 n^4 + 7137 n^5 + 
   1642 n^6)\\
 q_4(n)&:=1536 n^3 (56 + 420 n + 1510 n^2 + 2535 n^3 + 2351 n^4 + 706 n^5)\\
 q_5(n)&:=2048 n^4 (252 + 1260 n + 3485 n^2 + 3495 n^3 + 2554 n^4)\\
 q_6(n)&:=8192 n^5 (252 + 756 n + 1653 n^2 + 669 n^3 + 512 n^4)\\
  q_7(n)&:=393216 n^6 (14 + 21 n + 39 n^2)\\
   q_8(n)&:=524288 n^7 (18 + 9 n + 16 n^2)\\
   q_9(n)&:=9437184 n^8\\
   q_{10}(n)&:=4194304 n^9\\
\end{split}
\end{align}
Taking~$n=100$, one can check, via Sturm's algorithm~\cite{sturm}, that the polynomial 
\begin{eqnarray}
\mathfrak{p}(n,|\hat{\omega}|):=\sum_{i=1}^{9}p_i(n)|\hat{\omega}|^{i-1}
\end{eqnarray}
has two real roots in~$|\hat{\omega}|\in (0,1)$. Taking~$|\hat{\omega}|=\frac{3}{10}$ and~$|\hat{\omega}|=\frac{8}{10}$, one can check~$\mathfrak{p}(100,\frac{3}{10})<0$ and~${\mathfrak{p}(100,\frac{8}{10})<0}$. Hence, $E_0\leq \frac{E(u_T)}{||u_T||^2_{L^2(\mathbb{R})}}<-\frac{1}{4000}<0$ for all $|\hat{\omega}|\in [\frac{3}{10},\frac{8}{10}]$. 
\end{proof}
\subsection{Proof of Proposition~\ref{propG}}
To prove proposition~\ref{propG}, one can clearly reformulate as follows:
 \begin{prop}\label{propdirect}
For all~$|\hat{\omega}|\in[\frac{3}{10},\frac{8}{10}]$, there exists a~$C^{\infty}(\mathbb{R})$ solution~$\tilde{\mathfrak{h}}$ to the Schr\"odinger equation~(\ref{ODEF}) with~$\hat{\mu}> \frac{1}{20\sqrt{10}}>0$, and in the language of proposition~\ref{Asymp2},~$\tilde{k}_2=0$ and~$\tilde{c}_1=0$.
 \end{prop}
 \begin{proof}
 By proposition~\ref{propvar}, corollary~\ref{corm} and proposition~\ref{h2proof}, for all $\omega\in [\frac{3}{10},\frac{8}{10}]$, there exists a minimiser $u\in H^1(\mathbb{R})$ with $||u||_{L^2(\mathbb{R})}=1$ such that 
 \begin{eqnarray}
 E(u)=E_0:=\inf\Big\{\langle \nabla v,\nabla v\rangle_{L^2(\mathbb{R})}+\langle V v,v \rangle_{L^2(\mathbb{R})}:{v\in H^1(\mathbb{R})},{||v||_{L^2(\mathbb{R})}=1}\Big\}
 \end{eqnarray}
 with $V$ as defined in equation~(\ref{pot}). Moreover, by proposition~\ref{h2proof}, $E_0<-\frac{1}{4000}<0$. \\
 
By standard Euler--Lagrange methods (see theorem~3.21 and example~3.22 in~\cite{RIND}),~$u$ will weakly solve the ODE
\begin{eqnarray}
-\frac{d^2u}{dx_*^2}+V(x_*)u=-\hat{\mu}^2 u
\end{eqnarray}
with $-\hat{\mu}^2=E_0$. From proposition~\ref{h2proof}, $\hat{\mu}^2=-E_0>\frac{1}{4000}$. Hence, for all $|\hat{\omega}|\in[\frac{3}{10},\frac{8}{10}]$, there exists a weak solution $u\in H^{1}(\mathbb{R})$ to the Schr\"odinger equation~(\ref{ODEF}) with $||u||_{L^2(\mathbb{R})}=1$ and $\hat{\mu}=\sqrt{-E_0}>\frac{1}{20\sqrt{10}}$.\\

From the regularity theorem~\ref{Schreg}, any~$u\in H^1(\mathbb{R})$ which weakly solves the Sch\"odinger equation~(\ref{ODEF}) is in fact smooth. Therefore, for all $|\hat{\omega}|\in[\frac{3}{10},\frac{8}{10}]$, there exists a solution $u\in C^{\infty}(\mathbb{R})$ to the Schr\"odinger equation~(\ref{ODEF}) with $\hat{\mu}=\sqrt{-E_0}>\frac{1}{20\sqrt{10}}$.\\

 To verify the boundary conditions of $u$, recall by proposition~\ref{Asymp2} the solution $u$ can be expressed, in the bases associated to $r=2M$ and $r\rightarrow \infty$, as
 \begin{align}
 u&=\tilde{k}_1\tilde{\mathfrak{h}}^{2M,+}+\tilde{k}_2\tilde{\mathfrak{h}}^{2M,-}\\
 u&=\tilde{c}_1\tilde{\mathfrak{h}}^{\infty,+}+\tilde{c}_2\tilde{\mathfrak{h}}^{\infty,-}
 \end{align}
 with $\tilde{k}_1,\tilde{k}_2,\tilde{c}_1,\tilde{c}_2\in \mathbb{R}$.
Note that
\begin{align}
\int_{-\infty}^{0}|\tilde{\mathfrak{h}}^{2M,-}|^2dx_*=\int_{1}^{\frac{3}{2}}|\tilde{\mathfrak{h}}^{2M,-}|^2\frac{x}{x-1}dx&=\infty,
\end{align}
whilst 
\begin{align}
\int_{-\infty}^{0}|\tilde{\mathfrak{h}}^{2M,+}|^2+|\Delta_{x_*}\mathfrak{h}^{2M,+}|^2dx_*=\int_{1}^{\frac{3}{2}}\Big(|\tilde{\mathfrak{h}}^{2M,+}|^2+\Big|\frac{x-1}{x}\Delta_{x}\mathfrak{h}^{2M,+}\Big|^2\Big)\frac{x}{x-1}dx&<\infty.
\end{align}
Similarly, for $X_*>0$ sufficently large
\begin{align}
 \int_{X_*}^{\infty}|\tilde{\mathfrak{h}}^{\infty,+}|^2dx_*=\int_{x(X_*)}^{\infty}|\tilde{\mathfrak{h}}^{\infty,+}|^2\frac{x}{x-1}dx&=\infty,
 \end{align}
 whilst 
\begin{align}
\int_{X_*}^{\infty}|\tilde{\mathfrak{h}}^{\infty,-}|^2+|\Delta_{x_*}\mathfrak{h}^{\infty,-}|^2dx_*=\int_{x(X_*)}^{\infty}\Big(|\tilde{\mathfrak{h}}^{\infty,-}|^2+\Big|\frac{x-1}{x}\Delta_{x}\mathfrak{h}^{\infty,-}\Big|^2\Big)\frac{x}{x-1}dx&<\infty.
\end{align}
Therefore, since $u\in H^1(\mathbb{R})$, the solution $u$, in the language of proposition~\ref{Asymp2}, must have~$\tilde{k}_2=0$ and~$\tilde{c}_1=0$. \\

Therefore, taking~$\tilde{\mathfrak{h}}=u$ and~$|\hat{\omega}|\in [\frac{3}{10},\frac{8}{10}]$ gives a~$C^{\infty}(\mathbb{R})$ solution to the Schr\"odinger equation~(\ref{ODEF}) with~$\hat{\mu}>\frac{1}{20\sqrt{10}}>0$,~$\tilde{k}_2=0$ and~$\tilde{c}_1=0$. 
 \end{proof}

\pagebreak
\appendices
\section{Christoffel and Riemann Tensor Components for the 5D Schwarzschild Black String}\label{ChR}
To compute $\Box_g h_{ab}$ one requires the Christoffel symbols and the Riemann tensor components; the non-zero Christoffel symboles are listed below:
\begin{eqnarray}
\Gamma^t_{tr}&=&\frac{M}{r(r-2M)},\\
\Gamma^r_{tt}&=&\frac{M(r-2M)}{r^3},\qquad \Gamma^r_{rr}=\frac{-M}{r(r-2M)},\qquad \Gamma^r_{{\theta}{\theta}}=(2M-r),\qquad \Gamma^r_{{\phi}{\phi}}=(2M-r)\sin^2\theta,\\
 \Gamma^{{\theta}}_{r{\theta}}&=&\frac{1}{r}=\Gamma^{\phi}_{r{\phi}},\qquad \Gamma^{{\theta}}_{{\phi}{\phi}}=-\sin\theta\cos\theta,\qquad \Gamma^{\phi}_{{\theta}{\phi}}=\cot\theta.
\end{eqnarray}
The others are obtained from symmetry of lower indices. Note, ${R^z}_{\mu\alpha\beta}={R^\mu}_{z\alpha\beta}={R^{\mu}}_{\alpha z\beta}={R^{\mu}}_{\alpha\beta z}=0$. So the Riemann tensor components that are relevant are the ones with spacetime indices $\mu\in \{0,...,3\}$ which are just the usual Schwarzschild Riemann tensor components; the non-zero ones are listed below for completeness, 
\begin{eqnarray}
{R^t}_{rtr}=\frac{2M}{r^2(r-2M)},\qquad {R^t}_{{\theta}t{\theta}}&=&-\frac{M}{r},\qquad {R^t}_{{\phi}t{\phi}}=-\frac{M\sin^2\theta}{r},\\
{R^r}_{trt}=-\frac{2M(r-2M)}{r^4},\qquad {R^{r}}_{{\theta}r{\theta}}&=&- \frac{M}{r},\qquad {R^r}_{{\phi}r{\phi}}=-\frac{M\sin^2\theta}{r},\\
{R^{\theta}}_{t{\theta}t}=\frac{M (r-2M)}{r^4},\qquad {R^{\theta}}_{r{\theta}r}&=&-\frac{M }{r^2 (r-2M)},\qquad {R^{\theta}}_{{\phi}{\theta}{\phi}}= 
 \frac{2 M \sin^2\theta}{r},\\
{R^{{\phi}}}_{t{\phi}t}=\frac{M(r-2M)}{r^4},\qquad {R^{{\phi}}}_{r{\phi}r}&=&-\frac{M}{r^2(r-2M)},\qquad {R^{\phi}}_{{\theta}{\phi}{\theta}}=\frac{2M}{r}.
\end{eqnarray}
Any others can be found from the ${R^a}_{b(cd)}=0$ symmetry. 
\pagebreak

\section{Singularities in Second Order ODE}\label{Sing}
This section is heavily based on the book of Olver~\cite{Oliver}. In particular, see chapter 5 sections 4 and 5 and chapter 7 section 2. 
\begin{definition}[Ordinary Point/Regular Singularity/Irregular Singularity]\label{RegSingDef}
Let~$p$ and~$q$ be meromorphic functions on a subset of $\mathbb{C}$. Consider the linear $2^{\mathrm{nd}}$ order ODE
\begin{eqnarray}
\frac{d^2f}{dz^2}+p(z)\frac{df}{dz}+q(z)f=0.\label{ODESING}
\end{eqnarray}
Then $z_0\in \mathbb{C}$ is an ordinary point of this differential equation if both $p(z)$ and $q(z)$ are analytic there. If~$z_0$ is not an ordinary point and both 
\begin{eqnarray}
(z-z_0)p(z)\quad \text{and} \quad (z-z_0)^2q(z)
\end{eqnarray}
are analytic at $z_0$ then $z_0$ is a regular singularity, otherwise $z_0$ is an irregular singularity. 
\end{definition}
\begin{remark}
The singular behavior of $z=\infty$ is determined by making the change of variables $\tilde{z}=\frac{1}{z}$ in the ODE (\ref{ODESING}). This case will be considered explicitly in section~\ref{irrengsingsec}.
\end{remark}
In the following, general results for ODE are presented. 
\subsection{Regular Singularities}
\noindent In this paper solutions of a second order ODE in a neighbourhood $|z-z_0|<r$ of a regular singular point are required. The classical method is to search for a convergent series solution in such a neighbourhood. 
\begin{definition}[Indicial Equation]
Let~$p$ and~$q$ be meromorphic functions on a subset of $\mathbb{C}$. Consider the following $2^{nd}$-order ODE with a regular singularity at~${z_0\in \mathbb{C}}$
\begin{eqnarray}
\frac{d^2f}{dz^2}(z)+p(z)\frac{df}{dz}(z)+q(z)f(z)=0.\label{ODEREG}
\end{eqnarray}
Assume that there exist a convergent power series, 
\begin{eqnarray}
(z-z_0)p(z)=\sum_{k=0}^{\infty}p_k(z-z_0)^k,\qquad (z-z_0)^2q(z)=\sum_{k=0}^{\infty}q_k(z-z_0)^k\qquad \forall |z-z_0|<r.\label{assumregode}
\end{eqnarray}
The indicial equation is defined as
\begin{eqnarray}
I(\alpha):=\alpha(\alpha-1)+p_0\alpha+q_0=0.
\end{eqnarray}
\end{definition}
\begin{remark}
The indicial equation arises by considering the a solution of the form $f(z)=(z-z_0)^{\alpha}$ to the ODE
\begin{eqnarray}
\frac{d^2f}{dz^2}(z)+\frac{p_0}{z-z_0}\frac{df}{dz}(z)+\frac{q_0}{(z-z_0)^2}f(z)=0.\label{ODEREGAPP}
\end{eqnarray}
The ODE~(\ref{ODEREGAPP}) is the leading order approximation of the ODE~(\ref{ODEREG}). The function $f(z)=(z-z_0)^{\alpha}$ solves the ODE~(\ref{ODEREGAPP}) if the $\alpha$ satisfies the indicial equation.
\end{remark}
\noindent The following two theorems deal with the asymptotic behavior of solutions in the neighbourhood of a regular singularity.
\begin{theorem}[Frobenius]\label{ThE1}
Let~$p$ and~$q$ be meromorphic functions on a subset of $\mathbb{C}$. Consider the following~$2^{nd}$-order ODE with a regular singularity at~${z_0\in \mathbb{C}}$
\begin{eqnarray}
\frac{d^2f}{dz^2}(z)+p(z)\frac{df}{dz}(z)+q(z)f(z)=0,\label{RSODE}
\end{eqnarray}
where
\begin{eqnarray}
(z-z_0)p(z)=\sum_{k=0}^{\infty}p_k(z-z_0)^k,\qquad (z-z_0)^2q(z)=\sum_{k=0}^{\infty}q_k(z-z_0)^k\label{seriespq}
\end{eqnarray}
converge for all $|z-z_0|<r$, where $r>0$. Let $\alpha_{\pm}$ be the two roots of the indicial equation. Suppose further that $\alpha_-\neq \alpha_++s$, where $s\in \mathbb{Z}$. Then there exists a basis of solution to the ODE~(\ref{RSODE}) of the form
\begin{eqnarray}
f^+(z)=(z-z_0)^{\alpha_+}\sum_{k=0}^{\infty}a_k^+(z-z_0)^k,\qquad f^-(z)=(z-z_0)^{\alpha_-}\sum_{k=0}^{\infty}a_k^-(z-z_0)^k\label{formalps}
\end{eqnarray}
 where these series converge for all $z$ such that $|z-z_0|<r$. Moreover, $a_k^+$ and $a_k^-$ can be calculated recursively by the formula
 \begin{eqnarray}
 I(\alpha_{\pm}+k)a^{\pm}_k+(1-\delta_{k,0})\sum_{s=0}^{k-1}\big((\alpha_{\pm}+s)p_{k-s}+q_{k-s}\big)a^{\pm}_s=0.\label{vanseries}
 \end{eqnarray}
\end{theorem}
\begin{remark}
If the roots of the indicial equation do not differ by an integer then theorem~\ref{ThE1} gives a basis of solutions for the ODE in a neighbourhood of the singular point. Equation~(\ref{vanseries}) determines the coefficients of the series expansion recursively from an arbitrarily assigned $a_0\neq 0$, which can be taken to be $1$. This process runs into difficulty if, and only if, the two roots differ by a positive integer. To see this, let $\alpha_+$ be the root of the indicial equation with largest real part, the other root is then $\alpha_+-s$ for some $s\in \mathbb{Z}_+$. Then since $I((\alpha_+-s)+s)=0$ one cannot determine $a_s$ via equation~(\ref{vanseries}) for this power series. In this case, one solution can be found with the above method by taking the root of the indicial equation with largest real part. 
\end{remark}
\noindent The following theorem investigates the case where the roots differ by an integer. Let $\alpha_+$ be the root of the indicial equation with largest real part, the other root is then $\alpha_+-s$ for some $s\in \mathbb{Z}_+\cup \{0\}$. 
\begin{theorem}\label{ThE2}
Consider the ODE~(\ref{RSODE}) as in theorem~\ref{ThE1} again satisfying~(\ref{seriespq}).
Let~$\alpha_+$ and~${\alpha_-=\alpha_{+}-N}$, with~$N\in \mathbb{Z}_+\cup \{0\}$, be roots of the indicial equation. Then there exists a basis of solutions of the form
\begin{eqnarray}
f^+(z)=(z-z_0)^{\alpha_+}\sum_{k=0}^{\infty}a_k^+(z-z_0)^k,\quad f^-(z)=(z-z_0)^{\gamma}\sum_{k=0}^{\infty}a^+_k(z-z_0)^k+C_Nf^+(z)\ln (z-z_0)
\end{eqnarray}
with $\gamma=\alpha_++1$ if $N=0$ and $\gamma=\beta_-$ if $N\neq 0$, where these power series are convergent for all $z$ such that $|z-z_0|<r$. Moreover, the coefficents $a_k^+$, $a_k^-$ and $C_N$ can be calculated recursively.
\end{theorem}
\subsection{Irregular Singulities}\label{irrengsingsec}
This section summaries the key result for constructing a basis of solutions to the ODE~(\ref{HZ}) associated to $r\rightarrow \infty$. (The results presented can in fact be applied to any irregular singular point of an ODE~(\ref{ODESING}) since without loss of generality, the irregular singularity can be assumed to be at infinity after a change of coordinates.) The following definition makes precise the notion of a irregular singularity at infinity.
\begin{definition}[Irregular Singularity at Infinity]\label{irregsinginf}
Let~$p$ and~$q$ be meromorphic functions on a subset of $\mathbb{C}$ which includes the set~$\{z\in \mathbb{C}:|z|>a\}$. Consider the following~$2^{nd}$-order ODE
\begin{eqnarray}
\frac{d^2f}{dz^2}+p(z)\frac{df}{dz}+q(z)f=0.\label{IRSODE}
\end{eqnarray}
Assume for $|z|>a$, $p$ and $q$ may be expanded as convergent power series
\begin{eqnarray}
p(z)=\sum_{n=0}^{\infty}\frac{p_n}{z^n},\qquad q(z)=\sum_{n=0}^{\infty}\frac{q_n}{z^n}
\end{eqnarray}
The ODE~(\ref{IRSODE}) has an irregular singular point at infinity if one of $p_0$, $q_0$ and $q_1$ do not vanish. 
\end{definition}
The main theorem~\ref{ThE3} of this section can be motivated by the following discussion. Consider a formal power series
\begin{eqnarray}
w=e^{\lambda z}z^{\mu}\sum_{n=0}^{\infty}\frac{a_n}{z^n}\label{IS4}.
\end{eqnarray}
Substituting the expansions into the ODE and equating coefficients yields
\begin{eqnarray}
\lambda^2+p_0\lambda+q_0&=&0\label{IS1}\\
(p_0+2\lambda)\mu&=&-(p_1\lambda+q_1)\label{IS2}
\end{eqnarray}
and
\begin{eqnarray}
(p_0+2\lambda)na_n&=&(n-\mu)(n-1-\mu)a_{n-1}+\sum_{j=1}^{n}(\lambda p_{j+1}+q_{j+1}-(j-n-\mu)p_j)a_{n-j}\label{IS3}.
\end{eqnarray}
Now, equation~(\ref{IS1}) has two roots
\begin{eqnarray}
\lambda_{\pm}=\frac{1}{2}\Big(-p_0\pm\sqrt{p_0^2-4g_0}\Big).\label{lambdapm}
\end{eqnarray}
These give rise to
\begin{eqnarray}
\mu_{\pm}=-\frac{p_1\lambda_{\pm}+q_1}{p_0+2\lambda_{\pm}}.\label{mupm}
\end{eqnarray}
The two values of $a_0$, $a_{0}^{\pm}$ can be, without loss of generality, set to $1$ and the higher order coefficients determined iteratively from equation~(\ref{IS3}) unless one is in the exceptional case where $p_0^2=4g_0$ (for further information on this case see section~1.3 of chapter 7 in~\cite{Oliver}). The issue that arises is that in most cases the formal series solution~(\ref{IS4}) does not converge. However, the following theorem characterises when~(\ref{IS4}) provides an asymptotic expansion for the solution for sufficently large $|z|$. 

\begin{theorem}\label{ThE3}
Let $p(z)$ and $q(z)$ be meromorphic functions with convergent series expansions
\begin{eqnarray}
p(z)=\sum_{n=0}^{\infty}\frac{p_n}{z^n},\qquad q(z)=\sum_{n=0}^{\infty}\frac{q_n}{z^n}
\end{eqnarray}
for $|z|>a$ with $p_0^2\neq 4q_0$. Then the second order ODE
\begin{eqnarray}
\frac{d^2f}{dz^2}+p(z)\frac{df}{dz}+q(z)f=0\label{IS6}
\end{eqnarray}
has unique solutions $f^{\pm}(z)$, such that in the regions
\begin{eqnarray}
\begin{cases}
\{|z|>a\}\cap\{|\mathrm{Arg}((\lambda_--\lambda_+)z)|\leq \pi\}\quad (\text{for }f^+)\\
 \{|z|>a\}\cap\{|\mathrm{Arg}((\lambda_+-\lambda_-)z)|\leq \pi\}\quad (\text{for }f^-)\label{Regions}
\end{cases}
\end{eqnarray}
of the complex plane, $f^{\pm}$ is holomorphic, where $\lambda_{\pm}$ and $\mu_{\pm}$ are defined in equations~(\ref{lambdapm}) and~(\ref{mupm}). Moreover, for all $N>1$, $f^{\pm}(z)$ satisfies
\begin{eqnarray}
f^{\pm}(z)= e^{\lambda_{\pm}z}z^{\mu_{\pm}}\Big(\sum_{n=0}^{N-1}\frac{a^{\pm}_n}{z^n}+\mathcal{O}\Big(\frac{1}{z^{N}}\Big)\Big)
\end{eqnarray}
in the regions given in equation~(\ref{Regions}).
\end{theorem}
\pagebreak
\section{Tranformation to Schr\"odinger Form}\label{SchT}
\begin{prop} \label{Sch}Consider the second order homogeneous linear ODE
\begin{eqnarray}
\frac{d^2u}{dr^2}+p(r)\frac{du}{dr}+q(r)u=0,\qquad p,q\in C^1(I),\qquad I\subset\mathbb{R}
\label{ODE1}.
\end{eqnarray}
Suppose that there exists a sufficiently regular coordinate transformation $s(r)$ and a function $w(r)$ such that 
\begin{eqnarray}
\frac{dw}{dr}+\frac{1}{2}\Big(\frac{1}{\big(\frac{ds}{dr}\big)}\frac{d^2s}{dr^2}+p\Big)w=0.
\end{eqnarray}
Then the ODE~(\ref{ODE1}) can be reduced to the form
\begin{eqnarray}
-\frac{d^2z}{ds^2}(s)+V(s)z(s)=0,
\end{eqnarray}
with
\begin{eqnarray}
V(s)=\frac{1}{2\big(\frac{ds}{dr}\big)^2}\bigg(\frac{dp}{dr}-\frac{3}{2\big(\frac{ds}{dr}\big)^2}\Big(\frac{d^2s}{dr^2}\Big)^2+\frac{1}{\big(\frac{ds}{dr}\big)}\frac{d^3s}{dr^3}+\frac{p^2}{2}-2g\bigg).
\end{eqnarray}
\end{prop}
\begin{proof}
The proof is a straight-forward calculation. Take $u(s)=w(s)z(s)$, then
\begin{eqnarray}
\Big(\frac{ds}{dr}\Big)^2w\frac{d^2z}{ds^2}+\bigg(2\Big(\frac{ds}{dr}\Big)^2\frac{dw}{ds}+w\frac{d^2s}{dr^2}+pw\frac{ds}{dr}\bigg)\frac{dz}{ds}+\bigg(\Big(\frac{ds}{dr}\Big)^2\frac{d^2w}{ds^2}+\frac{dw}{ds}\frac{d^2s}{dr^2}+p\frac{dw}{ds}\frac{ds}{dr}+qw\bigg)z=0.\nonumber
\end{eqnarray}
To reduce this to symmetric form one can set
\begin{eqnarray}
2\Big(\frac{ds}{dr}\Big)^2\frac{dw}{ds}+w\frac{d^2s}{dr^2}+pw\frac{ds}{dr}=0,
\end{eqnarray}
which is equivalent to $w(r)$ satisfying
\begin{eqnarray}
\frac{dw}{dr}+\frac{1}{2}\Big(\frac{1}{\big(\frac{ds}{dr}\big)}\frac{d^2s}{dr^2}+p\Big)w=0.
\end{eqnarray}
Hence
\begin{eqnarray}
\frac{d^2w}{dr^2}=-\frac{1}{2}\bigg(\frac{df}{dr}-\frac{1}{\big(\frac{ds}{dr}\big)^2}\Big(\frac{d^2s}{dr^2}\Big)^2+\frac{1}{\big(\frac{ds}{dr}\big)}\frac{d^3s}{dr^3}-\frac{1}{2}\bigg(\frac{1}{\big(\frac{ds}{dr}\big)}\frac{d^2s}{dr^2}+p\bigg)^2\bigg)w.
\end{eqnarray}
Notice the last term in the ODE for $z$ reduces to
\begin{eqnarray}
\frac{d^2w}{dr^2}+p\frac{dw}{dr}+qw.
\end{eqnarray}
Reducing this with the expressions for the derivatives of $w$ gives the potential for $-\frac{d^2z}{ds^2}+V(s)z=0$ as
\begin{eqnarray}
V(s)&=&\frac{1}{2\big(\frac{ds}{dr}\big)^2}\bigg(\frac{dp}{dr}-\frac{3}{2\big(\frac{ds}{dr}\big)^2}\Big(\frac{d^2s}{dr^2}\Big)^2+\frac{1}{\big(\frac{ds}{dr}\big)}\frac{d^3s}{dr^3}+\frac{p^2}{2}-2q\bigg).
\end{eqnarray}
\end{proof}
\begin{remark}
Applying this to $s=r_*(r)=r+2M\log|r-2M|$ gives
\begin{eqnarray}
V(r(r_*))=\frac{(r-2M)^2}{2r^2}\bigg(\frac{df}{dr}+\frac{2M(2r-3M)}{r^2(r-2M)^2}+\frac{p^2}{2}-2q\bigg).\label{PotT}
\end{eqnarray}
\end{remark}
\pagebreak

\section{Useful Results From Analysis}\label{analysis}
\subsection{Sobolev Embedding}
\begin{theorem}[Local Compactness of the $H^s$ Sobolev Injection]\label{lc}
 Let $d\geq 1$, $s>0$ and
 \begin{eqnarray}
 p_c=\begin{cases}
 \frac{2d}{d-2s}\qquad s<\frac{d}{2}\\
 \infty\qquad \text{otherwise}.
 \end{cases}
 \end{eqnarray}
 Then the embedding $H^s(\mathbb{R}^d)\hookrightarrow L^p_{\mathrm{loc}}(\mathbb{R}^d)$ is compact $\forall 1\leq p< p_c$. In otherwords, for $(f_n)_n\subset H^s(\mathbb{R}^d)$ bounded, there exists $f\in H^s(\mathbb{R}^d)$ and a subsequence $(f_{n_m})_m$ such that
 \begin{eqnarray}
 &f_{n_m}\rightharpoonup f&\qquad H^s(\mathbb{R}^d),\\
 &f_{n_m}\rightarrow f&\qquad L^p_{\mathrm{loc}}(\mathbb{R}^d)\quad \forall 1\leq p<p_c.
 \end{eqnarray}
 \end{theorem}
 \begin{proof}
 This result can be found in any text on Sobolev spaces, for example Brezis~\cite{BPDE}. 
 \end{proof}
 \subsection{The Multiplication Operator is Compact}
 \begin{prop}\label{compactness}
Let $q\in C^0(\mathbb{R}^n,\mathbb{R})$ with $\lim_{|x|\rightarrow \infty}q(x)=0$ and $s>0$. Then $M_q:u\rightarrow qu$ is a compact operator from $H^s(\mathbb{R}^n,\mathbb{R})$ to $L^2(\mathbb{R}^n,\mathbb{R})$.
 \end{prop}
 \begin{proof}
\noindent The function $q$ is continuous and decays, hence it is bounded. Let $\epsilon>0$, then, by assumption, $\exists R>0$ such that
 \begin{eqnarray}
 |q(x)|\leq \epsilon \qquad \text{if }|x|\geq R.
 \end{eqnarray}
 Define, $\chi_R:\mathbb{R}\rightarrow \mathbb{R}$ smooth by
\begin{eqnarray}
\chi_R(x)=\begin{cases}
1 \qquad |x|\leq R\\
0\qquad |x|\geq R+1.
\end{cases}
\end{eqnarray}
Let $(f_n)_n\subset H^s(\mathbb{R}^n,\mathbb{R})$ be bounded, so local compactness of the Sobolev embedding (theorem~\ref{lc}) gives convergence in $H^s(\mathbb{R}^n,\mathbb{R})$ and weak convergence in $L^2_{\mathrm{loc}}(\mathbb{R}^n,\mathbb{R})$ up to a subsequence. Let the limit be $f\in H^s(\mathbb{R}^n,\mathbb{R})$. Therefore, 
\begin{eqnarray}
||\chi_Rqf_{m_n}-\chi_R qf||^2_{L^2(\mathbb{R}^n)}&=&||\chi_Rqf_{m_n}-\chi_R qf||^2_{L^2(B_{R+1}(0))}\\
&\leq & C\sup_{x\in \mathbb{R}}|q(x)|^2||f_{m_n}-f||^2_{L^2(B_{R+1}(0))}\leq \epsilon^2.
\end{eqnarray}
Further, consider the set $S_{R}:=\{\chi_{R}qf:f\in H^s(\mathbb{R}^n,\mathbb{R}),||f||_{H^s(\mathbb{R}^n)}\leq 1\}$. Then
\begin{eqnarray}
||(1-\chi_R)q f||_{L^2(\mathbb{R}^n)}\leq \epsilon^2||f||_{L^2(\mathbb{R}^n)}\leq \epsilon^2.
\end{eqnarray}
Hence, $S_{\infty}$ is within an $\epsilon$-neighbourhood of $S_R$, which is compact, therefore $S_{\infty}$ is compact. By the characterization of compactness through weak convergence, $qf_m\rightarrow q f$ in $L^2(\mathbb{R}^n,\mathbb{R})$ up to a subsequence.
 \end{proof}
 \subsection{A Regularity Result}
\begin{theorem}[Regularity for the Schr\"odinger Equation]\label{Schreg}
Let $u\in H^1(\mathbb{R})$ be a weak solution of the equation $(-\Delta+V)u=\lambda u$ where $V$ is a measurable function and $\lambda\in \mathbb{C}$. Then, if $V\in C^{\infty}(\Omega)$ with $\Omega \subset \mathbb{R}$ open, not necessarily bounded, then $u\in C^{\infty}(\Omega)$ also.  
\end{theorem}
\begin{proof}
Reed and Simon volume II page 55~\cite{RSII}. Note one can argue this from standard elliptic regularity results and Sobolev embeddings. In this paper, only the one-dimensional case of this is applied, which is completely elementary.
\end{proof}

\pagebreak

\section{A Result on Stability in Spherical Gauge}\label{bounds}
This section contains a few technical results on where the instability may lie in frequency space. This helped guide the search for a suitable test function and the subsequent instability. 
\begin{prop}\label{poly}
Consider the quartic polynomial
\begin{eqnarray}
P(x)=ax^4+bx^3+cx^2+dx+e.
\end{eqnarray}
Let $\Delta$ denote its discriminant and define
\begin{eqnarray}
\Delta_0=64a^3e-16a^2c^2+16ab^2c-16a^2bd-3b^4.
\end{eqnarray}
If $\Delta<0$, then $P(x)$ has two distinct real roots and two complex conjugate roots with non-zero imaginary part. If $\Delta>0$ and $\Delta_0>0$, then there are two pairs of complex conjugate roots with non-zero imaginary part. 
\end{prop}
\begin{proof}
See reference~\cite{poly}.
\end{proof}
\begin{prop}[Regions of Stability in Frequency Space]
Let $\mu> 0$ and $\omega\neq 0$. There does not exist a solution $\mathfrak{h}$ of the ODE~(\ref{HZ}) with $c_1=0$, $k_2=0$ and $\hat{\omega}\in \mathbb{R} \setminus (-\sqrt{2},\sqrt{2})$ or $\hat{\mu}\geq \frac{3}{8}\sqrt{\frac{3}{2}}$. 
\end{prop}
\begin{proof}
From proposition~\ref{Asymp}, the admissible boundary conditions for the solution are $\mathfrak{h}(r)=k_1\mathfrak{h}^{2M,+}(r)$ at the future event horizon and $\mathfrak{h}(r)=c_2\mathfrak{h}^{\infty,-}_z(r)$ at spacelike infinity. Without loss of generality, take $k_1>0$. Now, since the solution must decay exponentially towards infinity, there must be maxima $a\in (1,\infty)$. At such a point, one has
\begin{eqnarray}
\frac{d^2\mathfrak{h}}{dr^2}(a)=\frac{a(\hat{\mu}^2 a+\hat{\omega}^2(\hat{\mu}^2a^4-2a+2)+\hat{\omega}^4 a^3(a-1))}{(\hat{\omega}^2a^3+1)(a-1)^2}\mathfrak{h}(a),
\end{eqnarray} 
with $\mathfrak{h}(a)>0$. To derive a contradition, one must have
\begin{eqnarray}
\frac{a(\hat{\mu}^2 a+\hat{\omega}^2(\hat{\mu}^2a^4-2a+2)+\hat{\omega}^4 a^3(a-1))}{(\hat{\omega}^2a^3+1)(a-1)^2}>0.
\end{eqnarray} A sufficient condition for the numerator to be positive is 
\begin{eqnarray}
\hat{\mu}^2a^4-2a+2\geq 0.
\end{eqnarray}
This has discriminant
\begin{eqnarray}
\Delta=16\hat{\mu}^4(128\hat{\mu}^2-27),\qquad \Delta_0=128\hat{\mu}^2.
\end{eqnarray}
Hence, if $\hat{\mu}^2>\frac{27}{128}$, then there are no real roots. Thus, because the polynomial is positive at a point, say $a=1$, it is positive everywhere. If $\Delta=0$, there is a double real root and two complex conjugate roots. The real roots can only occur at a stationary point of the polynomial and therefore the polynomial cannot be negative anywhere. Since all other terms in the numerator are positive, the prefactor of $\mathfrak{h}$ also is. Hence, there can be no solution with the conditions $k_2=0$ and $c_1=0$ if $\hat{\mu}\geq \frac{3}{8}\sqrt{\frac{3}{2}}$. \\

\noindent Another sufficient condition for positivity of the numerator is
\begin{eqnarray}
\hat{\omega}^2a^3-2\geq 0.
\end{eqnarray}
This polynomial has a single real root at $a=\big(\frac{2}{\hat{\omega}^2}\big)^{\frac{1}{3}}$. For positivity on $a\in (1,\infty)$, one requires $\frac{2}{\hat{\omega}^2}\leq 1$ or $\hat{\omega}^2\geq 2$. Note that if $\hat{\mu}=0$ then this is precisely the polynomial that governs positivity. Hence, this bound for $\hat{\omega}$ is sharp.
\end{proof}

\begin{remark}
By an almost identical argument one can make the bound for $\hat{\mu}$ even sharper and show that $\hat{\mu}<\frac{1}{4}$ and $\hat{\mu}\leq \sqrt{2}|\hat{\omega}|$. 
\end{remark}
\pagebreak
\bibliography{BS}
\bibliographystyle{ieeetr}
 \end{document}